%% file: spaa2021-KClique.tex
\documentclass[sigconf,nonacm]{acmart}

\def\BibTeX{{\rm B\kern-.05em{\sc i\kern-.025em b}\kern-.08emT\kern-.1667em\lower.7ex\hbox{E}\kern-.125emX}}
    
\usepackage{graphicx,subcaption}
\usepackage{todonotes}

\input{extrapackages}
\input{macrosetup}

\input{acronyms}

\usepackage{arydshln}

\newtheorem{observation}{Observation}


\begin{document}

\fancyhead{}

\title{Parallel Algorithms for Finding Large Cliques in Sparse Graphs}


\author{Lukas Gianinazzi$^1$, Maciej Besta$^1$, Yannick Schaffner$^2$, Torsten Hoefler$^1$}
       \affiliation{\vspace{0.3em}$^1$Department of Computer Science, ETH Zurich;
       {$^2$}Department of Mathematics, ETH Zurich
}

\begin{abstract}
We present a parallel $k$-clique listing algorithm with improved work bounds (for the same depth) in sparse graphs with low degeneracy or arboricity. 
We achieve this by introducing and analyzing a new pruning criterion for a backtracking search.
Our algorithm has better asymptotic performance, especially for larger cliques (when $k$ is not constant), where we avoid the straightforwardly exponential runtime growth with respect to the clique size. In particular, for cliques that are a constant factor smaller than the graph's degeneracy, the work improvement is an exponential factor in the clique size compared to previous results. Moreover, we present a low-depth approximation to the community degeneracy (which can be arbitrarily smaller than the degeneracy). This approximation enables a low depth clique listing algorithm whose runtime is parameterized by the community degeneracy. 
\end{abstract}


\maketitle

\section{Introduction}
\normalem
Finding large cliques has many applications in the social sciences,
bioinformatics, computational chemistry, and others~\cite{tomita2011efficient,
rehman2012graph, lee2010survey, shao2012managing, tang2010graph,
aggarwal2010managing, jiang2013survey, cook2006mining, besta2021graphminesuite, besta2021sisa}. As the problem is
NP-hard and remains hard even when parameterized by the size of the clique
$k$~\cite{Downey2013}, it makes sense to consider special families of graphs
for which the problem is tractable. If the work of an algorithm
  can be written as $\bigo{f(P)N^c}$ (for some function $f$ of the \emph{parameters} $P$ and some
  polynomial $N^c$ of the input size $N$), the problem is
  \emph{fixed parameter tractable} (FPT) with respect to the parameters $P$~\cite{Downey1995}.
  The clique problem is FPT with respect to certain \emph{structural sparsity
  parameters} of graphs.

As a classic example family of sparse graphs, planar graphs do not contain
cliques with more than $4$ vertices, and all cliques in a planar graph can be
listed in linear time in the size of the graph~\cite{PAPADIMITRIOU1981131}.
Planar graphs are members of a more general class of sparse graphs for which
every subgraph has a low-degree vertex. This notion of sparsity is called
\emph{degeneracy}. A closely related family of sparse graphs is defined by
being decomposable into a small number of forests. The \emph{arboricity}
measures the number of such forests required for a given graph. Many real-world
graphs are sparse~\cite{besta2019demystifying, besta2015accelerating,
besta2019slim, besta2019practice, besta2017push, besta2018log} and have a low degeneracy and
arboricity~\cite{Danisch2018, Shi2020, besta2018survey, besta2020high,
DBLP:journals/kais/ShinEF18}.

In this work, we focus on improving the work and depth of FPT algorithms to list all $k$-cliques in these sparse graphs, improving both work and depth significantly as a function of the clique size $k$.
Our approach utilizes the notion of \emph{edge communities} - vertices with whom an edge forms a triangle. 
We further introduce a notion of \emph{relevant edges}, that effectively upper
bounds the community's size. This enables us to reject edges that
cannot be a part of a \kclique, pruning unnecessary recursive calls. We show that
this pruning criterion asymptotically reduces the work compared to previous
approaches.

\subsection{Preliminaries}

We consider a (directed or undirected) graph $G=(V, E)$ with $n$ vertices $V$ and $m$ edges $E$. We will generally assume the graph is connected and therefore $m=\Omega(n)$. To orient a graph \emph{by a total order}, direct its edges from the endpoint lower in the total order to the endpoint higher in the total order. A graph that has been oriented by a total order is acyclic by construction.
The subgraph of a graph $G$ induced by a vertex set $V'$ is $G[V']$.
The \emph{neighbors} of a vertex $u$ in the graph $G$ are the set of all vertices in $V$ that are connected to $u$ by edges in $E$. In a directed graph, the \emph{out-neighbors} are $\outneighbor{u}{G} = \{v\in V\ \mid\ (u,v)\in E\}$, and the \emph{in-neighbors} are $\inneighbor{v}{G} = \{u\in V\ \mid\ (u,v)\in E\}$. If the graph $G$ is clear from the context, we write $\neighbors{u}$, $\outneighbors{u}$, and $\inneighbors{u}$ for short.
	
In an undirected graph, the \emph{community} $C_G(u,v)=C_G(e)$ of an edge $e=\{u,v\}$ is the intersection of the neighbors of its endpoints. In a directed graph, the \emph{community} $C_G(u,v)=C_G(e)$ of an edge $e=(u,v)$ is the intersection of the \emph{out}-neighbors of $u$ and the \emph{in}-neighbors of $v$. If the graph $G$ is clear from the context, we write $C(u,v)$ and $\community{e}$ for short.

\paragraph{Sparse Graphs}

A graph $G=(V,E)$ is $s$-\emph{degenerate}, if in all induced subgraphs $H$ of
$G$, there is a vertex with degree at most $s$ \cite{Lick1970}. The
\emph{degeneracy} of a graph is the smallest $s$ such that the graph is
$s$-degenerate. Note that an $s$-degenerate graph can have an unbounded maximum
degree. For example, the star graph is $1$-degenerate but has a maximum degree
of $n-1$. An $s$-degenerate graph can be oriented such that every vertex has at
most $s$ out-neighbors (greedily remove a vertex with the smallest degree in
the remaining subgraph). The order in which this greedy procedure removes
vertices is a \emph{degeneracy order}.

A graph $G=(V, E)$ is $\sigma$-\emph{community degenerate}, if every
(non-edgeless) subgraph $G'$ has an edge $e$ with $\cardinality{ C_{G'}(e) }
\le \sigma$. The \emph{community degeneracy} of a graph is the smallest value
$\sigma$, such that the graph is $\sigma$-community degenerate
\cite{buchanan2014solving}. The community degeneracy is strictly smaller than
the degeneracy: $\sigma < s$ and there are families of graphs where the
community degeneracy is asymptotically smaller. For example, the
$d$-dimensional hypercube has degeneracy $s=d$, but community degeneracy
$\sigma=0$.
Moreover, the community degeneracy can be 1 for graphs where the degeneracy is
$\Theta(n)$ and there are $\Theta(n)$ triangles. This is the case for a
graph which has the complete bipartite graph as a subgraph, with $n/2$ vertices in each part,
and additionally one part of the graph induces a line graph on $n/2$
vertices.
Buchanan et al.~\cite{buchanan2014solving} also show empirically that the community
degeneracy is significantly smaller (27\%-80\%) than the degeneracy for some
real-world graphs (when the degeneracy is significantly larger than the clique
number).
	
The smallest number of forests into which a graph can be decomposed (meaning that every edge is in exactly one of the forests) is the \emph{arboricity} $\alpha$ of a graph. It is closely related to the degeneracy, in particular $\alpha \leq s < 2\alpha$~\cite{Nash-Williams1961}.

\paragraph{$k$-Cliques}

An induced subgraph of $G$ that is a complete graph with $k$-vertices is a $k$-clique. It follows from the definition that an $s$-degenerate graph does not contain any $(s+2)$-clique and a graph with arboricity $\alpha$ does not contain any $(2\alpha+1)$-clique (using $\alpha > s/2$). In this paper, we generally assume $k\geq 4$.

Deciding the size of a maximum clique is a classic NP-hard problem~\cite{GareyNP-complete} and so is deciding if the graph has a \kclique. The problem of finding a clique of size $k$ remains hard when parameterized by $k$: it is W[1]-hard~\cite{Downey2013}. However, the problem becomes FPT for graphs when parameterized by degeneracy or arboricity~\cite{Chiba1985}.

\paragraph{Model of Computation} 
Our execution model is a shared-memory parallel computer with concurrent reads and exclusive writes (CREW PRAM)~\cite{Reif:1993:SPA:562546}. 
We express the performance of our algorithms in the \emph{work/depth model}.
 The \emph{work} is the total number of elementary operations performed by all processors in any algorithm execution. The \emph{depth} is the length of a critical path. An algorithm with work $W$ and depth $D$ can be scheduled on a CREW PRAM with $p$ processors so that it takes $O(W/p + D)$ time steps~\cite{Blelloch:1996:PPA:227234.227246, Reif:1993:SPA:562546}.

\subsection{Related Work}

Chiba and Nishizeki~\citep{Chiba1985} presented a result on $k$-clique listing for low arboricity graphs. Their algorithm takes $O(m\alpha^{k-2})$ work.  
Danish et al.~\cite{Danisch2018} gave improved bounds in terms of the degeneracy. Their algorithm does $O(m (\frac{s}{2})^{k-2})$ work. Since $\alpha\leq s<2\alpha$, this is never worse than Chiba and Nishizeki's work bound, but can be faster by a term exponential in $k$ for graphs where the degeneracy is close to the arboricity. The algorithm has a depth of $O(n)$, stemming from how the graph is directed with a degeneracy order.
Recently, Shi et al.~\cite{Shi2020} gave a variant on the algorithm that uses a fast parallel approximation to the degeneracy order (in addition to providing some improvements in the data structure used to represent the graph during the recursive search). When $k$ is constant (which they assume), this approximation does not impact the runtime. However, the work increases by a factor exponential in $k$ because of this approximation.
These \kclique listing algorithms are not well suited to finding very large cliques, as their bounds become \emph{super-exponential} in $s$ (and $\alpha$) when the size of the cliques is $k=\Omega(s)$. This is far from optimal, as a graph with degeneracy $s$ has at most $(n-s+1)2^{s}$ cliques overall~\cite{DBLP:journals/gc/Wood07}.

Previous \kclique listing algorithms~\cite{Chiba1985,Danisch2018,Shi2020}
use a backtracking scheme combined with a clique-growing pattern, where a
subgraph is expanded into larger and larger cliques, until a \kclique is found
or determined to be unreachable. A common optimization~\cite{Danisch2018,Shi2020} is to \emph{orient} the input graph. In \kclique
counting, one usually orients a graph with a total order on the vertices
\cite{Danisch2018,Shi2020}. A popular choice of order is the degeneracy-order
as it assigns low out degrees and can be computed in linear time
\cite{Matula1983}.
Other works on $k$-cliques focus on reordering
heuristics~\cite{li2020ordering}, using GPUs~\cite{almasri2021k}.

Listing \emph{maximal} cliques is also a subject of numerous
works~\cite{bouchitte2002listing, cheng2011finding, modani2008large,
cheng2012fast, schmidt2009scalable}.
For example, Eppstein~\cite{Eppstein2010} presents a variant of Bron-Kerbosch's
algorithm~\cite{DBLP:journals/cacm/BronK73,DBLP:journals/tcs/TomitaTT06} to
find \emph{maximal} cliques. It takes $\bigo{sn 3^{s/3}}$ time, which is close
to the $(n-s)3^{s/3}$ lower bound~\cite{Eppstein2010} on the number of maximal
cliques in an $s$-degenerate graph.
A \emph{maximum} clique can be computed in $\bigo{2^{0.276n}}$ time from the relationship to maximum independent sets~\cite{DBLP:journals/jal/Robson86}.

Many other variants of clique listing exist, such as diversified top-$k$ cliques~\cite{yuan2016diversified}, the $k$-clique densest subgraph problem~\cite{DBLP:conf/www/Tsourakakis15a, mitzenmacher2015scalable}, and densest subgraph discovery~\cite{galbrun2016top}.

Converting a graph to a directed, acyclic graph gives an effective way to
assign cliques to certain vertices: A clique can \eg be associated with its
topologically smallest or largest vertex. This way, double counting can be
avoided and the directions reduce the size of the search space.

\subsection{Our Contributions}

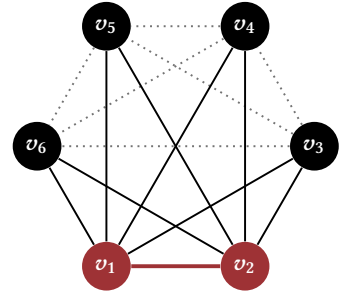
\begin{figure}
	\centering
	\resizebox{0.52\linewidth}{!}{%
			\boldmath
		\begin{tikzpicture}
			\node[current] (v1) at ({2*cos(5*60-60)},{2*sin(5*60-60)}) {$v_1$};
			\node[current] (v2) at ({2*cos(0*60-60)},{2*sin(0*60-60)}) {$v_2$};
			\node[current, fill=ccandidate1] (v3) at ({2*cos(1*60-60)},{2*sin(1*60-60)}) {$v_3$};
			\node[current, fill=ccandidate2] (v4) at ({2*cos(2*60-60)},{2*sin(2*60-60)}) {$v_4$};
			\node[current, fill=ccandidate3] (v5) at (({2*cos(3*60-60)},{2*sin(3*60-60)}) {$v_5$};
			\node[current, fill=ccandidate4] (v6) at ({2*cos(4*60-60)},{2*sin(4*60-60)}) {$v_6$};

			\foreach \from/\to in {
				v3/v4, v3/v5, v3/v6,
				v4/v5, v4/v6,
				v5/v6}
			\draw[edge] (\from) -- (\to);
			
			\draw[ecurthk] (v1) -- (v2);
			
			\draw[ecurrent,ccandidate1] (v1) -- (v3);
			\draw[ecurrent,ccandidate1] (v2) -- (v3);
			
			\draw[ecurrent,ccandidate2] (v1) -- (v4);
			\draw[ecurrent,ccandidate2] (v2) -- (v4);
			
			\draw[ecurrent,ccandidate3] (v1) -- (v5);
			\draw[ecurrent,ccandidate3] (v2) -- (v5);
			
			\draw[ecurrent,ccandidate4] (v1) -- (v6);
			\draw[ecurrent,ccandidate4] (v2) -- (v6);
		\end{tikzpicture}
	}
	\caption{To support a $k$-clique, an edge must have $k-2$ triangles that contain it. In other words, its community must have a size of $k-2$. In the example, the community of the edge $\{v_1, v_2\}$ contains all the other vertices $\{v_3, v_4, v_5, v_6\}$. Hence, it could potentially support a $6$-clique. Indeed, the edge $\{v_1, v_2\}$ does support a $6$-clique. } \label{fig:clique-triangles-example}
\vspace{-1em}
\end{figure}

We present an algorithm that provides, for the same depth, improved work bounds
for $k$-clique counting in $s$-degenerate graphs. In particular, when
  $k=\Omega(s)$, the improvement grows exponentially with $k$. 
%
%
Considering the case where $k$ is not constant is important because that is when the problem becomes hard.
See
\Cref{tab:bounds} for a comparison of our bounds compared to previous
$k$-clique listing results in sparse graphs. 
Our algorithm differs in two relevant aspects from previous approaches~\cite{Chiba1985,Danisch2018,Shi2020}.

First, existing approaches use the degree of a vertex to decide whether the neighborhood of a vertex \emph{can} contain a smaller clique. In contrast, our algorithm looks at edges and their \emph{triangles} to grow the clique: Each edge of a \kclique participates in (at least) $k-2$ triangles. See \Cref{fig:clique-triangles-example} for an example. 
Our algorithm preprocesses the graph such that each triangle $(a,b,c)$ is \emph{supported} by exactly one of its edges. If edge $(a,c)$ supports triangle $(a,b,c)$, we say that vertex $b$ is in the \emph{community} of edge $(a,c)$.
To find a \kclique, we only need to consider edges that support at least $k-2$ triangles. We add this edge to the clique and recursively search for a $(k-2)$-clique in the subgraph induced by the community of that edge. 

Second, we observe that we can use the number of vertices ordered between the endpoints of an edge as a proxy for the number of triangles it supports, as illustrated in \Cref{fig:clique-triangles-example-direction}. This observation allows us to exclude edges from the search space because they cannot support a \kclique, without actually looking at their triangles in every step. Moreover, this simplification enables us to upper bound our algorithm's work, improving on the state-of-the-art for $k$-clique counting in $s$-degenerate graphs.

Pruning edges by excluding those which do not have enough vertices ordered between their endpoints is what is responsible for the $\Theta\left(\left(\frac{1}{1-k/s}\right)^k\right)$ factor improvement in the work. 

We also discuss different ways to orient and preprocess the graph. These graph orientations leads to three variants with different work/depth tradeoffs, as seen in \Cref{tab:bounds}.

Finally, we show clique counting results parameterized by the community degeneracy. Our algorithm uses a total order \emph{on the edges} in addition to a total order on the vertices to get better results for graphs where $\sigma < s-1$. We again get three variants with different work/depth tradeoffs. To achieve the sublinear-depth variants, we present a novel low-depth approximation algorithm for the community degeneracy.

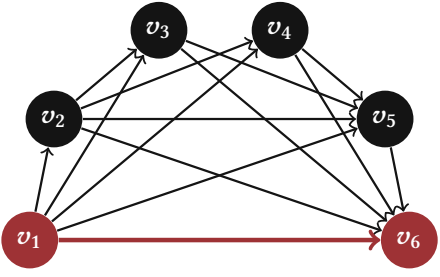
\begin{figure}
	\centering
	\resizebox{0.7\linewidth}{!}{%
	
		\resizebox{0.55\textwidth}{!}{%
			\boldmath
			\begin{tikzpicture}[->]
				
				\node[current] (v1) at(0.4, 0) {$v_1$};
				\node[candidate] (v2) at (0.7,1.5) {$v_2$};
				\node[candidate] (v3) at (2,2.6) {$v_3$};
				\node[candidate] (v4) at (3.5,2.6) {$v_4$};
				\node[candidate] (v5) at (4.8,1.5) {$v_5$};
				\node[current] (v6) at (5.1,0) {$v_6$};
				
				\foreach \from/\to in {
					v1/v2, v1/v3, v1/v4, v1/v5,
					v2/v3, v2/v4, v2/v5, v2/v6,
					v3/v5, v3/v6,
					v4/v5, v4/v6,
					v5/v6}
				\draw[ecandidate] (\from) -- (\to);
				
				
				\foreach \from/\to in {v1/v6}
				\draw[ecurthk] (\from) -- (\to);
			\end{tikzpicture}
		}
	}
\caption{If the graph is oriented by a total order, we can use the number of vertices ordered between the endpoints of an edge as a proxy to determine if it can support a $k$-clique. If there are less than $k-2$ vertices ordered between the endpoints, the edge cannot support a $k$-clique. In the example, only the edge $(v_1, v_6)$ could support a $6$-clique using this pruning rule. Therefore, we call the edge \emph{relevant}. However, the graph only contains two $5$-cliques and no $6$-clique because there is no edge $(v_3, v_4)$.}\label{fig:clique-triangles-example-direction}
%
%
\end{figure}

\renewcommand{\arraystretch}{1.8}
\begin{table}[t]
	\centering
	\small
	\tabcolsep=0mm
	\begin{tabular}{lcc} %
		\toprule
		& Work & Depth \\
		\midrule
		Chiba/Nishizeki~\cite{Chiba1985} \hspace{-2.2em} & $\bigo{m \alpha^{k-2}}$ & $\bigo{m \alpha^{k-2}}$ \\
		 Danish et al.~\cite{Danisch2018} & $\bigo{k m \left(\frac{s}{2}\right)^{k-2}}$ & $\bigo{n + \log^{2} n}$ \\
		Shi et al.~\cite{Shi2020} &  $\bigo{m (s(1+\epsilon))^{k-2}}^{\dag}$ & $\bigo{k \log n + \log^2 n}^{\star}$  \\
		\multicolumn{3}{l}{\underline{\textbf{Our Results for Degeneracy} } \S \ref{sec:clique-listing}: } \\
		\emph{Best Work}  \S \ref{sec:orientation}   & $\bigo{k m \left(\frac{s+3-k}{2}\right)^{k-2}}$ &  $\bigo{n + k \log n}$ \\
		\emph{Hybrid} \S \ref{sec:order-hybrid}  & $\bigo{k n s \left(\frac{s+3-k}{2}\right)^{k-2}}$ &  $\bigo{s + k \log n + \log ^2 n}$ \\
		\emph{Best Depth} \S \ref{sec:orientation}   & $\bigo{k m \left( \frac{s(2+\epsilon) +3-k}{2}\right)^{k-2}}$ &  $\bigo{k \log n + \log^2 n}$ \\
		\multicolumn{3}{l}{\underline{\textbf{Our Results for Community Degeneracy}} \S \ref{sec:clique-listing}: } \\
		\emph{Best Work} \S \ref{sec:commdeg-order} & $\bigo{ms + k m \left(\frac{\sigma+4-k}{2}\right)^{k-2}}$ &  $\bigo{n + k \log n}$ \\
		\emph{Hybrid} \S t\ref{sec:commdeg-order} & $\bigo{ms + k n \sigma \left(\frac{\sigma+4-k}{2}\right)^{k-2}}$ &  $\bigo{\sigma + k \log n + \log ^2 n}$ \\
		\emph{Best Depth} \S \ref{sec:commdeg-order}  & $\bigo{ms + k m \left(\frac{(3+\epsilon)\sigma+4-k}{2}\right)^{k-2}}$ &  $\bigo{ k \log n + \log^2 n}$ \\
\bottomrule
	\end{tabular}
	\vspace{1em}
	\caption{Bounds for listing all $k$-cliques in a graph ($k\geq 4$) with degeneracy $s$, community degeneracy $\sigma$, and arboricity $\alpha$. The parameter $\epsilon$ is some positive real \emph{constant}. Recall that $\alpha \leq s < 2\alpha$ and $k\leq \sigma+2 \leq s+1$. Bounds marked with ${}^{\dag}$ hold in expectation and those marked with ${}^{\star}$ hold with high probability in $n$.} \vspace{0em}\label{tab:bounds}
%
%
\end{table}

\section{Community-Centric Clique Listing} \label{sec:clique-listing}

We begin with an informal description of the ideas behind our \emph{community-centric clique listing} algorithm. Afterward, we state the algorithm and bound its work and depth. 
Let us start with a simple observation about \kcliques, on which we will gradually expand. \emph{Any induced subgraph on a clique is a clique as well.}
A slight reformulation of this observation is the primary building block of the algorithm by Chiba and Nishizeki~\cite{Chiba1985}, and subsequently, of Danisch et al.~\cite{Danisch2018} and Shi et al.~\cite{Shi2020}:
If a vertex's neighborhood contains a $(k-1)$-clique, then we have found a \kclique~\cite{Chiba1985}. This insight allows a simple backtracking formulation, where one recursively searches for smaller cliques in the local neighborhood of vertices. 
In the practical implementations \cite{Chiba1985,Danisch2018,Shi2020}, this formulation gives rise to a form of candidate growing. In each recursive call, the current candidate motif grows by one vertex to a larger clique. During the backtracking, the candidate motif shrinks again.

Observe that one is not limited to grow a motif by a single vertex in each step, but can extend the motif by any kind of clique in its neighborhood, such as an edge (\ie a $2$-clique).
Whereas each vertex has a fixed number of neighbors, each edge belongs to a fixed number of triangles. 
In a \kclique, each edge participates in $k-2$ triangles.
If the clique is oriented by a total order, we can refine the statement. The \emph{supporting edge} of a $k$-clique in a graph oriented by a total order is the edge connecting the first and last vertex in the total order. It holds that:
\begin{observation}\label{obs:trianglesSupportingEdge}
	The supporting edge $e$ of a \kclique oriented by a total order, has a community $\community{e}$ of size at least $k-2$. All other edges have a smaller community.
\end{observation}

In \Cref{sec:work-analysis}, we show the following bounds for our algorithm:
\begin{theorem}\label{thm:work}
	Let $G=(V,E)$ be a directed acyclic graph oriented by a total order, such that its largest community has size $\gamma$ and its largest out-degree is $\tilde s$. Then, the work performed by \Cref{alg:kcliqueListing} is
	\begin{displaymath}
		\bigo{m(\tilde s+\gamma^2) + k m \left(\frac{\gamma +4-k}{2}\right)^{k-2}} \enspace ,
	\end{displaymath}
	and its depth is 
	\begin{displaymath}
		\bigo{\log^2 n + k \log \gamma} \enspace .
	\end{displaymath}
\end{theorem}
In \Cref{sec:graph-orientation}, we discuss different approaches to to orient a graph, leading to the work/depth tradeoffs in \Cref{tab:bounds}. 

\subsection{Algorithm Construction}

\Cref{alg:kcliqueListing,alg:recursiveListing} summarize our clique listing algorithm. Input to \Cref{alg:kcliqueListing} are a graph \textsc{Dag} oriented by a total order and an integer $k>3$. 
First, construct the edge communities and sort them (this speeds up the later intersection operations). Second, loop over all edges and recurse on their communities (\cref{line:loop:kclique}) with the recursive clique counting \Cref{alg:recursiveListing}.

\begin{algorithm}
	\caption{Listing all $k$-cliques in the oriented graph \textsc{Dag}}\label{alg:kcliqueListing}
	\begin{algorithmic}[1]
		\State Build the communities and sort them\label{line:buildEdgeLists:kclique}
		\ParForAll{edges $e$ with at least $k-2$ triangles}\label{line:loop:kclique}
			\State\Call{RecursiveCount}{\textsc{Dag}$, \community{e}, k-2$}\label{line:recCall:kclique}
		\EndFor
	\end{algorithmic}
\end{algorithm}

The recursive procedure is shown in \Cref{alg:recursiveListing}. As input, it takes a directed graph \textsc{Dag}, a set of candidate vertices $I$ and a constant $c > 0$, stating the number of vertices still required to complete a \kclique. The recursion has two base cases: for $c=1$ (\cref{line:basecase1:recList}), where each vertex in the candidate set completes the partial motif to a \kclique and for $c=2$ (\cref{line:basecase2:recList}), where each connected pair of vertices completes the \kclique.

Given a set of vertices $I$ with a total order, the \emph{distance function} $\delta_I:I \times I \mapsto \{0, \dotsc, \cardinality{I-2}\}$ maps each pair of vertices $(u,v)$ to the number of elements in $V$ that are ordered between $u$ and $v$ in the total order. We keep the set $I$ in a sorted array. Then, for any two candidates $u$ and $v$, we compute $\delta_I(u, v)$ from their indices in $I$.

In the recursive case, iterate over all vertex pairs $(u,v)$, whose distance $\delta_I(u,v)$ to each other, given the total order and the candidate set $I$, allows them to support a \kclique (\cref{line:loop:recList}). For each such pair, probe the existence of the corresponding edge (\cref{line:existenceTest:recList}) and, if successful, intersect the candidate set with the community of the edge (\cref{line:intersection:recList}) to create a new candidate set $I'$. Then, make the next recursive call with the new candidate set $I'$, requiring $c-2$ vertices to complete a \kclique (\cref{line:recursion:recList}).

See \Cref{fig:6cliqueExample} for an example of an execution of \Cref{alg:kcliqueListing}.
\begin{algorithm}
	\caption{Recursively search for cliques of size $c$ in \textsc{Dag}[$I$].}\label{alg:recursiveListing}
	\begin{algorithmic}[1]
			\If{$c=1$}
				\State For every vertex $v$ in $I$, \Return a clique.
					\label{line:basecase1:recList}
			\ElsIf{$c=2$}
				\State For every edge $(u,v)$ in $G[I]$, \Return a clique.\label{line:basecase2:recList}
			\EndIf
			\ForAll{pairs $e=(u,v)\in I\times I$ s.t. $\delta_I(u,v) \ge c-2$}\label{line:loop:recList}
				\If{$e$ is an edge in \textsc{Dag}}\label{line:existenceTest:recList}
					\State $I'\gets I\cap \community{e}$ \label{line:intersection:recList}
					\State\Call{RecursiveCount}{\textsc{Dag}$,I', c-2$}\label{line:recursion:recList}
				\EndIf
			\EndFor
	\end{algorithmic}
	
\end{algorithm}

\begin{figure}[h]
		\captionsetup[subfigure]{}
		\centering
	\begin{subfigure}{1\linewidth}
		\centering
		\resizebox{0.42\textwidth}{!}{%
			\boldmath
			\begin{tikzpicture}[->]
				\node[current]   (v1) at ({-2*cos(0*60-60)},{2*sin(0*60-60)}) {$v_1$};
				\node[candidate] (v2) at ({-2*cos(1*60-60)},{2*sin(1*60-60)}) {$v_2$};
				\node[candidate] (v3) at ({-2*cos(2*60-60)},{2*sin(2*60-60)}) {$v_3$};
				\node[candidate] (v4) at ({-2*cos(3*60-60)},{2*sin(3*60-60)}) {$v_4$};
				\node[candidate] (v5) at ({-2*cos(4*60-60)},{2*sin(4*60-60)}) {$v_5$};
				\node[current]   (v6) at ({-2*cos(5*60-60)},{2*sin(5*60-60)}) {$v_6$};
				
				\foreach \from/\to in {
					v1/v2, v1/v3, v1/v4, v1/v5, v1/v6,
					v2/v3, v2/v4, v2/v5, v2/v6,
					v3/v5, v3/v6,
					v4/v5, v4/v6,
					v5/v6}
				\draw[ecandidate] (\from) -- (\to);

				\foreach \from/\to in {v1/v6}
				\draw[eold] (\from) -- (\to);
				
				\foreach \from/\to in {v1/v6}
				\draw[ecurrent] (\from) -- (\to);
			\end{tikzpicture}
		}
		
		\caption{Loop over all edges that can support a $4$-clique, \ie $4$ triangles, and call \Cref{alg:recursiveListing}. Only edge $(v_1,v_6)$ has a community of size at least $4$. For the call to \Cref{alg:recursiveListing}, $I=\{v_2, v_3, v_4, v_5\}$ and $c=4$. \vspace{0.6em}}
	\end{subfigure}
	\centering
	\begin{subfigure}{1\linewidth}
		\centering
		\resizebox{0.42\textwidth}{!}{%
			\boldmath
			\begin{tikzpicture}[->]
				\node[old]       (v1) at ({-2*cos(0*60-60)},{2*sin(0*60-60)}) {$v_1$};
				\node[current]   (v2) at ({-2*cos(1*60-60)},{2*sin(1*60-60)}) {$v_2$};
				\node[candidate] (v3) at ({-2*cos(2*60-60)},{2*sin(2*60-60)}) {$v_3$};
				\node[candidate] (v4) at ({-2*cos(3*60-60)},{2*sin(3*60-60)}) {$v_4$};
				\node[current]   (v5) at ({-2*cos(4*60-60)},{2*sin(4*60-60)}) {$v_5$};
				\node[old]       (v6) at ({-2*cos(5*60-60)},{2*sin(5*60-60)}) {$v_6$};
				
				\foreach \from/\to in {
					v1/v2, v1/v3, v1/v4, v1/v5, v1/v6,
					v2/v3, v2/v4, v2/v5, v2/v6,
					v3/v5, v3/v6,
					v4/v5, v4/v6,
					v5/v6}
				\draw[ecandidate] (\from) -- (\to);

				\foreach \from/\to in {v1/v6, v1/v2, v1/v5, v2/v5, v2/v6, v5/v6}
				\draw[eold] (\from) -- (\to);
				
				\foreach \from/\to in {v2/v5}
				\draw[ecurrent] (\from) -- (\to);
			\end{tikzpicture}
		}
		\caption{Iterate over all pairs in the candidate set with a distance of $2$ that form an edge. Only the edge $(v_2,v_5)$ has a distance of $2$.\vspace{0.6em}}
	\end{subfigure}
	\begin{subfigure}{1\linewidth}
		\centering
		\resizebox{0.42\textwidth}{!}{%
			\boldmath
			\begin{tikzpicture}[->]
				\node[old]        (v1) at ({-2*cos(0*60-60)},{2*sin(0*60-60)}) {$v_1$};
				\node[old]        (v2) at ({-2*cos(1*60-60)},{2*sin(1*60-60)}) {$v_2$};
				\node[current]    (v3) at ({-2*cos(2*60-60)},{2*sin(2*60-60)}) {$v_3$};
				\node[current]    (v4) at ({-2*cos(3*60-60)},{2*sin(3*60-60)}) {$v_4$};
				\node[old]        (v5) at ({-2*cos(4*60-60)},{2*sin(4*60-60)}) {$v_5$};
				\node[old]   	  (v6) at ({-2*cos(5*60-60)},{2*sin(5*60-60)}) {$v_6$};

				\foreach \from/\to in {v1/v6, v1/v2, v1/v5, v2/v5, v2/v6, v5/v6, v1/v3, v1/v4, v2/v3, v2/v4, v3/v5, v3/v6, v4/v5, v4/v6}
				\draw[eold] (\from) -- (\to);
				
			\end{tikzpicture}
		}
		
		\caption{Among the remaining candidates, there is only one pair: $(v_3,v_4)$. The edge between the pair does not exist. Hence \cref{alg:recursiveListing} aborts and reports that there is no $6$-clique.}
	\end{subfigure}
	\caption{Example of using \Cref{alg:kcliqueListing} to search (in vain) for a $6$-clique in the above graph. Black vertices belong to the candidate set for the next call to \Cref{alg:recursiveListing}, red vertices and edges are being considered for extending the clique, and grey vertices and edges have been already visited.  \vspace{0em} }\label{fig:6cliqueExample}
\end{figure}
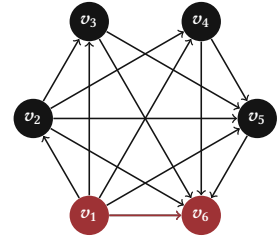
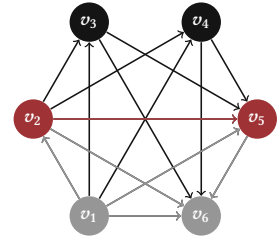
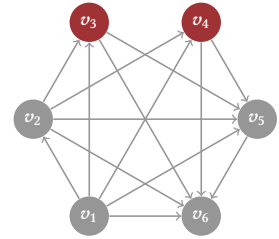

\begin{figure}[t]
		\centering
		\resizebox{0.53\linewidth}{!}{
			\boldmath
		\begin{tikzpicture}[->]
			\node[candidate] (v1) at(0.4, 0) {$v_1$};
			\node[candidate] (v2) at (0.7,1.5) {$v_2$};
			\node[candidate] (v3) at (2,2.6) {$v_3$};
			\node[candidate] (v4) at (3.5,2.6) {$v_4$};
			\node[candidate]       (v5) at (4.8,1.5) {$v_5$};
			\node[candidate]       (v6) at (5.1,0) {$v_6$};
			
			\draw[ecandidate] (v1) -- (v2);
			\draw[ecandidate] (v1) -- (v3);
			\draw[ecandidate] (v1) -- (v4);
			\draw[ecurthk] (v1) -- (v5);
			\draw[ecurthk] (v1) -- (v6);
			
			\draw[ecandidate] (v2) -- (v5);
			\draw[ecandidate] (v3) -- (v4);
			\draw[ecandidate] (v3) -- (v6);
			\draw[ecandidate] (v4) -- (v6);
		\end{tikzpicture}	
	}
	\vspace{-0em} 
	\caption{In the example graph, the edges relevant with respect to $3$ are $\relevantedge{G}{3}=\{(v_1, v_5), (v_1, v_6)\}$ (marked red). These are the edges that could potentially support a $5$-clique. The pairs $\relevantpair{G}{3}$ relevant with respect to $3$ also include $(v_2, v_6)$. \vspace{-0em}}  \label{fig:relevantEdgesPairs}
\end{figure}
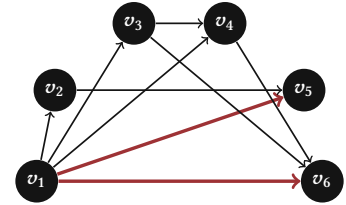

\subsection{Work/Depth Analysis}\label{sec:work-analysis}

Let us first bound the cost of the preprocessing in \Cref{alg:kcliqueListing}.
We can compute the triangles of a graph in $O(m\tilde s)$ work and $O(\log^2 n)$ depth~\cite{Shi2020}, which gives us the communities. Sorting the communities takes $\bigo{m \gamma \log \gamma}$ work and $\bigo{\log \gamma}$ depth~\cite{DBLP:journals/siamcomp/Cole88}.
We build for each edge $e$ an adjacency matrix of $G[\community{e}]$ and rename the vertices to be consecutive integers. Next, for each edge $e'$ in such a subgraph $G[\community{e}]$, we build a boolean indicator table for the vertices in its ``local'' community $C_{G[\community{e}]}(e')$.
This preprocessing speeds up the intersections and edge probing. It costs $\bigo{m \gamma^2}$ work. %

Before we bound the work and depth of the recursive \Cref{alg:recursiveListing}, we need some additional notation. 
Given a set of vertices $V$, a pair $(u,v)$ of vertices is \emph{relevant} with respect to $c\in \mathbb{N}_0$, if and only if $\delta_V(u,v) \ge c$, meaning that there are at least $c$ vertices ordered between $u$ and $v$. The set of \emph{relevant pairs} of $V$ with respect to $c$ and the total order of $V$ is $\relevantpair{V}{c}$. An edge $e=(u,v)$ is relevant with respect to $c$, if it forms a relevant pair with respect to $c$. Note that an edge that is not relevant with respect to $k-2$ cannot support a $k$-clique and that \Cref{alg:recursiveListing} only recurses on the relevant edges (with respect to $c-2$). The set of relevant edges with respect to $c$ in the graph $G=(V,E)$ is $\relevantedge{G}{c}$. See \Cref{fig:relevantEdgesPairs} for an illustration of the relevant pairs and edges.

Next, let us bound the work and depth of \Cref{alg:recursiveListing}.
For each relevant edge $e$, computing the intersection of $I$ with $\community{e}$ takes $\bigo{\log \gamma}$ depth and $\bigo{|\community{e}| + |I|}$ work: For each element in $I$ use the indicator table for $C(e)$ to test if it is in $C(e)$. Perform a parallel prefix sum~\cite{Reif:1993:SPA:562546} to gather the elements in the intersection. The work for edge probing is $\bigo{\relevantpair{I}{c-2}}$ and the depth is $\bigo{1}$. At the base cases, the work is $\bigo{k}$ per listed clique, and the depth is $\bigo{1}$. 

The recursion in \Cref{alg:recursiveListing} has depth $\lfloor \frac{k-2}{2}\rfloor$. Hence, the critical path has length $\bigo{\log n + k+ k \log(\gamma+1) }$. As long as $\gamma\neq 0$, the additive $k$ term is negligible. Next, we express the work of \Cref{alg:recursiveListing} recursively.

Let $W(c, I)$ be the work of \Cref{alg:recursiveListing} for candidate set $I$ and parameter $c$. We express the work recursively as follows:
\small
\begin{align*}
	W(1,I) &\le \bigo{k\enspace|I|}\\
	W(2,I) &\le \bigo{\cardinality{\relevantpair{I}{0}} + k\enspace\cardinality{\relevantedge{G[I]}{0}}}\\
	W(c,I) &\le \bigo{\cardinality{\relevantpair{I}{c-2}} \enspace + \sum_{e\in \relevantedge{G[I]}{c-2}} \cardinality{\community{e}}+\cardinality{I} + W(c-2, I \cap \community{e}) }
\end{align*}\normalsize
By telescoping this formula, we can see that the key to tackling this recursion is a bound on the following sum:
$$\sum_{e\in\relevantedge{G}{c}}\cardinality{\relevantedge{G[\community{e}]}{c-2}} \enspace .$$
The sum corresponds to the loop in the algorithm that iterates over the relevant edges. For each edge $e$, the term $\cardinality{\relevantedge{G[\community{e}]}{c-2}}$ is a bound on the number of relevant edges in the subgraph that the algorithm recurses on. 
In \Cref{sec:combinatorics}, we establish the following bounds on this quantity. 
\begin{lemma}\label{lem:relevantEdgeRecursion}
For a graph $G=(V,E)$ oriented by a total order and any $c\geq2$, we have that:
	\begin{align*}
		\sum_{e\in\relevantedge{G}{c}}\cardinality{\relevantedge{G[\community{e}]}{c-2}} \enspace & \le \sum_{e\in\relevantedge{G}{c}}\cardinality{\relevantpair{\community{e}}{c-2}} \\
		&\le \left(\frac{n - c}{2}\right)^2\cardinality{\relevantedge{G}{c}}.
	\end{align*}
\end{lemma}

We focus on bounding the total cost incurred for listing the cliques at the leaves of the recursion for $c\ge 2$. 
The proofs bounding the remaining cost (associated with the computation of intersections and edge testing) are similar and can be found in \Cref{sec:work-recursive-bounds}.
\begin{observation}\label{obs:level-function}
	The number of non-trivial recursive calls of  \Cref{alg:recursiveListing} (\ie excluding the base case $c=1$) is given by the function $r(c)= \left\lfloor\frac{c}{2}\right\rfloor$. Moreover, note that $r(c) + 1 = r(c+2)$.
\end{observation}

\noindent
	For some constant $a$, we define the following function to help unifying the analysis when $c$ is even or odd:
	\[b(c, i) \leq \begin{cases} 
      a\cdot k & \text{if $c$ is even} \\
      a\cdot k \cdot i & \text{else.}
  	  \end{cases}
	\]
	The constant $a$ is chosen such that the work to list a \kclique is at most $ak$ and the cost of listing $i=|I|$ \kcliques is at most $ a 
\cdot k \cdot i$.

\begin{lemma}\label{lem:listing-cost}
Let $L(c, I)$ be the work of the algorithm incurred at the leaves of the recursion in \Cref{alg:recursiveListing}. Then, we have for any $c\ge 2$,
$$
		L(c, I) \le \left(\frac{\cardinality{I}-c+2}{2}\right)^{2r(c-2)}\cardinality{\relevantedge{G[I]}{c-2}}\enspace  b (c, \cardinality{I}-2r(c)) \enspace . \label{eq:work:listing}
$$
\end{lemma}
\begin{proof}
The proof is by induction on $c$.
The base cases $c=2$ and $c=3$ are quickly verified. 
Assume now, that the inequality on $L(c', I)$ holds for all $c'\leq c$. We show that it also holds for $c+2$.
	\begin{align*}
		& \enspace L(c,I) \\
		\le  &\sum_{e\in\relevantedge{G[I]}{c-2}}L(c-2, I \cap \community{e}) \\
		\le &\sum_{e\in\relevantedge{G[I]}{c-4}}  \left(\frac{\cardinality{I \cap \community{e}}-c+4}{2}\right)^{2r(c-4)} 
	\cardinality{\relevantedge{G[I \cap \community{e}]}{c-4}}\\
		& \enspace \cdot  b(c-2, \cardinality{I \cap \community{e}}-2r(c-2)) \\%
		\le &\sum_{e\in\relevantedge{G[I]}{c-2}}  \left(\frac{\cardinality{I}-c+2}{2}\right)^{2r(c-4)}\cardinality{\relevantedge{G[I \cap \community{e}]}{c-4}} \\
		&\enspace \cdot  b(c-2, \cardinality{I}-2-2r(c-2)).
	\end{align*}
	Now, we apply \Cref{lem:relevantEdgeRecursion} (on the subgraph $G[I]$):
	\begin{align*}
		&\le \left(\frac{\cardinality{I}-c+2}{2}\right)^{2r(c-4)+2}\cardinality{\relevantedge{G[I]}{c-2}} \enspace b(c-2, \cardinality{I}-2-2r(c-2))
	\end{align*}
	Now, we use \Cref{obs:level-function} for both occurrences of the function $r$. Furthermore, since $c$ and $c-2$ are either both odd or both even, we can replace the $c-2$ in the call to $b$ with $c$:
	\begin{align*}
		L(c,I) &\le \left(\frac{\cardinality{I}-c+2}{2}\right)^{2r(c-2)}\cardinality{\relevantedge{G[I]}{c-2}} \enspace b(c, \cardinality{I}-2r(c)),
	\end{align*}	
	which is of the desired form and concludes the proof.
\end{proof}
See \Cref{sec:work-outer-loop} for a proof of the overall work incurred by \Cref{alg:kcliqueListing}, finishing the proof of \Cref{thm:work}.

\section{Combinatorics}\label{sec:combinatorics}

In this section, we prove \Cref{lem:listing-cost}, which is the key ingredient in the work proof. We begin by introducing the notation that we need for the proof. See also \Cref{fig:topologicalOrder} and \Cref{fig:relevantEdges} for an illustration of the relevant concepts.

\vspace{-1em}
\subsection{Additional Notation}

A vertex $u\in V$ is a \emph{relevant out-vertex} with respect to $c$, if there exists a vertex $v \in V$, such that $(u,v)$ forms a relevant pair with respect to $c$. The set of relevant out-vertices with respect to $c$ and the total order in $V$ is denoted $\relevantout{V}{c}$.
Similarly, a vertex $v\in V$ is a \emph{relevant in-vertex} with respect to $c$, if there exists a vertex $u \in V$, such that $(u,v)$ forms a relevant pair with respect to $c$.
The set of relevant in-vertices with respect to $c$ and the total order in $V$ is denoted $\relevantin{V}{c}$.
See \Cref{fig:topologicalOrder} for an example.

$\relevantedgeout{G}{c}$ are all relevant out-vertices with respect to $c$, that are part of a \emph{relevant edge}, not only a relevant pair. Similarly, $\relevantedgein{G,u}{c}$ are all relevant in-vertices with respect to $c$, that form a relevant edge together with  $u$, see \Cref{fig:relevantEdges}. Formally, we have:
$$\relevantedgeout{G}{c} = \{ u \enspace | \enspace \exists w:\; \enspace (u, w) \in \relevantedge{G}{c}\} \enspace ,$$
$$\relevantedgein{G,u}{c} = \{ v \enspace | \enspace  (u, v) \in \relevantedge{G}{c}\} \enspace .$$

For a set $S$ and a \emph{predicate} $P$, we write $S[P]$ to select the subset of $S$ that satisfies the predicate $P$, \eg	$V[<u]$ contains all vertices in $V$ smaller than $u$ in the order $<$.

\begin{figure}
\vspace{-1em}
		\centering
		\resizebox{0.75\linewidth}{!}{
			\boldmath
			\begin{tikzpicture}[->]
				\node[highLeft] (v1) at (-3,-3) {$v_1$};
				\node[highLeft] (v2) at (-1.5,-3) {$v_2$};
				\node[candidate] (v3) at (0,-3) {$v_3$};
				\node[candidate] (v4) at (1.5, -3) {$v_4$};
				\node[highRight] (v5) at (3, -3) {$v_5$};
				\node[highRight] (v6) at (4.5, -3) {$v_6$};
				
				\draw[ecurthk] (v1) to[out=25,in=155] (v5);
				\draw[ecurthk] (v1) to[out=-25,in=-155] (v6);
				\draw[ecurthk] (v2) to[out=25,in=155] (v6);
				
			\end{tikzpicture}
		}
\vspace{-1em}
		\caption{Ordered between the pairs $(v_1,v_5)$, $(v_1, v_6)$ and $(v_2,v_6)$ are at least $3$ other vertices. Hence, they are all relevant pairs with respect to $3$. That is, they constitute the set $\relevantpair{G}{3}$. %
		Each of the relevant pairs starts with one of the vertices in $\relevantout{V}{3}=\{v_1, v_2\}$ (marked orange) and ends with one of the vertices in $\relevantin{V}{3}=\{v_5, v_6\}$ (marked blue).} \label{fig:topologicalOrder}
\vspace{-1em}
\end{figure}
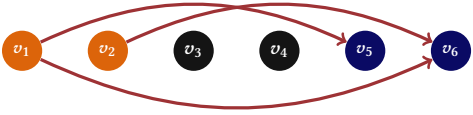

\begin{figure}
		\centering
		\resizebox{0.49\linewidth}{!}{
			\boldmath
		\begin{tikzpicture}[->]
			
			\node[highLeft] (v1) at(0.4, 0) {$v_1$};
			\node[candidate] (v2) at (0.7,1.5) {$v_2$};
			\node[candidate] (v3) at (2,2.6) {$v_3$};
			\node[candidate] (v4) at (3.5,2.6) {$v_4$};
			\node[highRight]       (v5) at (4.8,1.5) {$v_5$};
			\node[highRight]       (v6) at (5.1,0) {$v_6$};

			\draw[ecandidate] (v1) -- (v2);
			\draw[ecandidate] (v1) -- (v3);
			\draw[ecandidate] (v1) -- (v4);
			\draw[ecurthk] (v1) -- (v5);
			\draw[ecurthk] (v1) -- (v6);
			
			\draw[ecandidate] (v2) -- (v5);
			\draw[ecandidate] (v3) -- (v4);
			\draw[ecandidate] (v3) -- (v6);
			\draw[ecandidate] (v4) -- (v6);
		\end{tikzpicture}	
	}
\vspace{-1em}
	\caption{In the example graph $\relevantedge{G}{3}=\{(v_1,v_5),(v_1, v_6)\}$. Hence, the relevant out-vertices with respect to $3$ are $\relevantedgeout{G}{3} = \{v_1\}$ (marked orange). Each such vertex is the head of an edge that has distance $3$. The vertices $\relevantedgein{G,v_1}{3}=\{v_5, v_6\}$ are the tails of those edges (marked blue). }\label{fig:relevantEdges}
\vspace{-1em}
\end{figure}
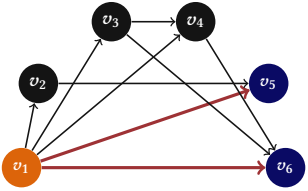

\subsection{Bounds on Relevant Edges and Pairs}

Let us begin with two simple observations on relevant pairs:
\begin{observation}\label{cor:countRelevantVertices}
	In a set of vertices $V$ with a total order, the number of relevant out- and in-vertices with respect to $c\in\mathbb{N}_0$ is:
	\begin{align*}
		\cardinality{\relevantout{V}{c}} &= \cardinality{\relevantin{V}{c}} = \cardinality{V} - (c+1) \enspace .
	\end{align*}
\end{observation}
\begin{proof}
	The last (respectively first) $(c+1)$ vertices in $V$ cannot be relevant with respect to $c$.
\end{proof}

\begin{observation}\label{cor:countRelevantPairs}
	In a set of vertices $V$ with a total order, the number of relevant pairs in $V$ with respect to $c \in\mathbb{N}_0$ is:
	\begin{align*}
		\cardinality{\relevantpair{V}{c}} &= \cardinality{\relevantout{V}{c}}\frac{\cardinality{V}-c}{2} = \cardinality{\relevantin{V}{c}}\frac{\cardinality{V}-c}{2} =  {{\cardinality{V}-c}\choose{2}} \enspace . 
	\end{align*}
\end{observation}

\begin{proof}
	The first relevant out-vertex forms relevant pairs with exactly $\cardinality{V}-(c+1)$ other vertices. The second out-vertex with $\cardinality{V}-(c+2)$ vertices, and so on, until the very last relevant out-vertex only participates in one relevant pair. The count of relevant pairs is thus a sum over a descending sequence from $\cardinality{V}-(c+1)$ down to $1$. Since there are $\cardinality{\relevantout{V}{c}}=\cardinality{V} - (c+1)$ terms, the sum can be computed with Gauss' sum formula. A symmetric argument relates the relevant in-vertices to the relevant pairs.
\end{proof}

We continue with a simpler (weaker) Lemma related to \Cref{lem:relevantEdgeRecursion}:
\begin{lemma}\label{lem:upperBoundSumRelevantInducedEdgesCommunityDegeneracy}
For a graph $G=(V,E)$ oriented by a total order where the largest community has size $\gamma$ and any $c\geq2$, we have that: 
	\begin{displaymath}
		\sum_{e\in \relevantedge{G}{c}}\cardinality{\relevantedge{G[\neighbors{e}]}{c-2}} \le {{\gamma-c+2}\choose{2}}\cardinality{\relevantedge{G}{c}} \enspace.
	\end{displaymath}
\end{lemma}
\begin{proof}
Every relevant edge is also a relevant pair, so we bound the number of relevant pairs  $\cardinality{\relevantpair{\community{e}}{c-2}}$ for each edge $e$. 	For any edge $e\in E$, by \Cref{cor:countRelevantPairs} and $\cardinality{\community{e}}\leq \gamma$, 
	\begin{align*}
		\cardinality{\relevantedge{\community{e}}{c-2}} \le \cardinality{\relevantpair{\community{e}}{c-2}} = {{\cardinality{\community{e}}-(c-2)}\choose{2}}.
	\end{align*}
\end{proof}
When applied to a graph with at most $\gamma+2$ vertices, \Cref{lem:listing-cost} provides a factor $2$ better bound than \Cref{lem:upperBoundSumRelevantInducedEdgesCommunityDegeneracy}. This factor is crucial for the work bound as it will be amplified exponentially with $k$ by the recursion. Now, let us turn to the proof of our main result on counting relevant edges.
\begin{proof}[Proof Of \Cref{lem:relevantEdgeRecursion}]
	For each term of the sum, we rewrite $\cardinality{\relevantpair{\community{e}}{c-2}}$ using \Cref{cor:countRelevantPairs}. Note that $\cardinality{\community{e}} \le n-2 $ to upper bound the term:
	\begin{align}
		&\sum_{e\in\relevantedge{G}{c}}\cardinality{\relevantpair{\community{e}}{c-2}} \\ \le \enspace
		&\frac{n -c}{2}\sum_{e\in\relevantedge{G}{c}}\cardinality{\relevantout{\community{e}}{c-2}}, \\
		\le \enspace &\frac{n -c}{2}\sum_{u\in\relevantedgeout{G}{c}}\sum_{v\in\relevantedgein{G,u}{c}}\cardinality{\relevantout{\community{u,v}}{c-2}}.
	\end{align}	
	Now, we relate the set of relevant in-vertices that are part of a relevant edge to a set of vertices that are part of a relevant pair: 
	
	\begin{equation}\label{eq:relevantedgeinTorelevantin}
		\relevantedgein{G,u}{c} \supseteq \relevantin{\{u\} \cup\outneighbor{u}{G}}{c} = \relevantin{\outneighbor{u}{G}}{c-1}
	\end{equation}
	Moreover, for every vertex $v$ that is in $\relevantedgein{G,u}{c}$ and is not in $\relevantin{\{u\} \cup\outneighbors{u}}{c}$,  the set $\relevantpair{\community{u,v}}{c-2}$ is empty: Consider a vertex in $\relevantedgein{G,u}{c}$ that is not in $\relevantin{\{u\} \cup\outneighbor{u}{G}}{c}$. Then, there are less than $c$ vertices ordered between $u$ and $v$ that are inside $\outneighbor{u}{G}$ (using that $u$ is smaller than its out-neighbors). Hence, the set $\outneighbor{u}{G}[<v]$ has size less than $c$. If the set $\relevantpair{\community{u,v}}{c-2}$ is not empty, then there is a pair $(u', v')$ in it with distance $\delta_{C(u,v)}(u', v')\geq c-2$. But this implies the set $C(u,v)$ has size at least $c$. Because $C(u,v)$ is a subset of $\outneighbor{u}{G}[<v]$ ( all vertices in $\inneighbors{v}$ come before $v$ in the total order) this contradicts the fact that $\outneighbor{u}{G}[<v]$ has size less than $c$.
	
	This means that the terms in the inner sum corresponding to vertices $v$ that are in $\relevantedgein{G,u}{c}$ but not in $\relevantin{\{u\} \cup\outneighbors{u}}{c}$ can be dropped without changing the sum.	
	Applying this observation and $\community{u,v} \subseteq \outneighbors{u}[<v]$, we continue:
	\begin{align}
		&\frac{n -c}{2}\sum_{u\in\relevantedgeout{G}{c}}\sum_{v\in\relevantedgein{G,u}{c}}\cardinality{\relevantout{\community{u,v}}{c-2}} \\
		= \enspace &\frac{n-c}{2}\sum_{u\in\relevantedgeout{G}{c}}\sum_{v\in\relevantin{\outneighbors{u}}{c-1}}\cardinality{\relevantout{\community{u,v}}{c-2}} \\
	 	\leq \enspace &\frac{n-c}{2}\sum_{u\in\relevantedgeout{G}{c}}\sum_{v\in\relevantin{\outneighbors{u}}{c-1}}\cardinality{\relevantout{\outneighbors{u}[<v]}{c-2}} 
	\end{align}
	Continue using $\relevantout{\outneighbors{u}[< v]}{c-2} =	\relevantout{\outneighbors{u}[\le v]}{c-1} \enspace $:
	\begin{align}
		= \enspace &\frac{n-c}{2}\sum_{u\in\relevantedgeout{G}{c}}\sum_{v\in\relevantin{\outneighbors{u}}{c-1}}\cardinality{\relevantout{\outneighbors{u}[\le v]}{c-1}} \enspace .
	\end{align}
	Above, the inner sum now counts all pairs in $\outneighbors{u}$, that are relevant with respect to $c-1$. We rewrite the sum using \Cref{cor:countRelevantPairs} applied to the subgraph of $G$ induced by $\outneighbors{u}$ (which has at most $n-1$ vertices):
	\begin{align}
		&\le \frac{n-c}{2}\sum_{u\in\relevantedgeout{G}{c}}\cardinality{\relevantin{\outneighbors{u}}{c-1}}\frac{n-1-(c-1)}{2} \\
		&\le \left(\frac{n-c}{2}\right)^2\sum_{u\in\relevantedgeout{G}{c}}\cardinality{\relevantedgein{G,u}{c}}, \\
		&\le \left(\frac{n-c}{2}\right)^2\cardinality{\relevantedge{G}{c}}.
	\end{align}
	In the second step, we upper bounded $|\relevantin{\outneighbors{u}}{c-1}|$ using \Cref{eq:relevantedgeinTorelevantin}. The sum now counts again over the edges, relevant with respect to $c$, which concludes the proof.
\end{proof}

\section{Graph Orientation}\label{sec:graph-orientation}

Our clique listing algorithm requires a vertex ordering such that when we orient the graph, its communities have a small size $\gamma$. We explore three different approaches. The first two approaches rely on existing results on degeneracy orders. The third hybrid approach uses a different outer loop to achieve a better work/depth tradeoff for large cliques $k$.

We also present how to get algorithms with bounds that are parametric in the community degeneracy. These algorithms require an order \emph{on the edges} that reduces the size of the subgraphs used in the recursive calls over the edges' communities.

\subsection{Computing a Degeneracy Order}\label{sec:orientation}
We can use previous results on degeneracy order to compute such an order, with different work/depth tradeoffs.
A greedy algorithm computes the degeneracy order in linear time, but has linear depth.
\begin{lemma}
	Computing a degeneracy order takes $O(m)$ work and $O(n)$ depth~\cite{Matula1983}.
\end{lemma}
Although no low depth algorithm is known to compute the degeneracy order, it can be approximated effectively in parallel. A  total order on the vertices such that the graph oriented by it has maximum out-degree $\beta s$ is a $\beta$-\emph{approximate degeneracy order}.
\begin{lemma}[Besta et al.\cite{Besta20HPCColoring}, Shi et al.~\cite{Shi2020}]
	Computing a $(2+\epsilon)$-approximate degeneracy order takes $O(m)$ work and $O(\log n \log_{1+\epsilon}n)$ depth.
\end{lemma}
Note that the size of the largest community is at most the maximum out-degree minus 1. Hence, using these vertex orders to orient the graph, we obtain two of the results in \Cref{tab:bounds} from \Cref{thm:work}, namely the best work and the best depth results for degeneracy. 

\subsection{The Hybrid Approach}\label{sec:order-hybrid}

While computing the exact degeneracy order provides the best work-bound in terms of $m, k$, and $s$, its depth is linear in $n$. Using an approximate degeneracy order yields a clique counting algorithm that is work-inefficient by a term that is exponential in $k$. This inefficiency persists even for very small $\epsilon$ because the approximation leads to a factor $2+\epsilon$ larger maximum out-degree in the worst case. Hence the factor $\Theta( (2+\epsilon)^k)$ work-inefficiency arises. 

Therefore, we present a hybrid approach that is only work-inefficient by a factor $\frac{ns}{m}\leq s$ and has a depth that reduces the term linear in $n$ to linear in $s$. The idea is that an approximate degeneracy order already guarantees that the subgraphs induced by the out-vertices have $O(s)$ vertices. It then remains to solve these subgraphs with the clique listing algorithm that uses an exact degeneracy order. The algorithm goes as follows:

Compute a $(2.5)$-approximate degeneracy order in parallel and orient the graph according to it. Then, for each vertex $v$, orient the graph $G[\outneighbors{v}]$ using an (exact) degeneracy order, then run \Cref{alg:recursiveListing} on $G[\outneighbors{v}]$.

This hybrid approach has $\bigo{kns \left(\frac{s+3-k}{2}\right)^{k-2}}$ work and depth $\bigo{s+k\log s + \log^2 n}$. For each vertex $v$, the induced subgraph $G[\outneighbors{v}]$ has $O(s|\outneighbors{v}|)$ edges and $O(s)$ vertices. The cost of building these subgraphs and their communities is $\bigo{s \outneighbors{v}}$ per vertex $v$ and $\bigo{ms}$ overall. For $k\geq 4$, the preprocessing cost does not dominate. Finally, apply the bounds from \Cref{thm:work} on each subgraph for $\gamma=s-1$. 

\subsection{Parameterizing by Community Degeneracy} \label{sec:commdeg-order}
Next, we present a preprocessing approach that uses a total order of the edges in addition to an order on the vertices. This preprocessing naturally leads to algorithms whose runtime is parameterized by the community degeneracy of the graph.

To reduce the runtime for the recursive clique listing algorithm, we want to reduce the size of the subgraphs induced by the communities of an edge. We observe that it suffices to recurse only on a subset of the community, namely we can introduce an order on the edges of the graph. Then, we only consider the community of an edge $e$ in the subgraph induced by edges ordered higher than $e$. In other words, these are the endpoints of the triangles supported by edges whose other edges come after the supporting edge. 
The approach is summarized in \Cref{alg:kcliqueListingCommunity} and goes as follows:

In addition to computing a total order on the vertices, first, compute \emph{a total order $\preceq$ on the edges}. Then, traverse the edges in this order and for each edge $e=(u,v)$ construct the candidate set $C_{(V, E[e\preceq])} (e)$, that is, the community of the edge $e$ in the subgraph of $G$ consisting of the edges ordered after $e$ in the total order. 
For each candidate set, run the recursive clique counting \Cref{alg:recursiveListing} with parameter $c=k-2$. 

To reduce the work of this algorithm, the chosen edge order is crucial. We minimize the maximum size of the candidate sets with a greedy algorithm: repeatedly choose and remove an edge that support the smallest number of triangles in the remaining subgraph. By definition of the community degeneracy $\sigma$, the largest subgraph constructed in this way has $\sigma$ vertices. Note that the order of the vertices used for the recursive clique listing does not influence the runtime and can be arbitrary (say by vertex id).
\begin{algorithm}
	\caption{Listing all $k$-cliques in graph $G$}\label{alg:kcliqueListingCommunity}
	\begin{algorithmic}[1]
		\State \textsc{Dag}$ \gets$ orient $G$ by a total vertex order
		\State Compute a total order $\preceq$ of the edges of $G$
		\ParForAll{edges $e=\{u,v\}$}\label{line:loop:kclique}
			\State $V' \gets C_{(V, E[e\preceq])} (e)$
			\State\Call{RecursiveCount}{\textsc{Dag}$, V', k-2$}\label{line:recCall:kcliqueComm}
		\EndFor
	\end{algorithmic}
\end{algorithm}

Observe that the set of candidates $V'$ constructed for edge $e$ is a subset of the community $C_G(e)$ of $e$ in the input graph. 
Moreover, each $k$-clique is contained in at least one of the constructed subgraphs, namely the one which corresponds to the edge in the clique lowest in the community degeneracy order. The total order on the vertices ensures that each clique is reported exactly once.

\begin{theorem}\label{thm:comm-deg-algo-work}
Let $\gamma$ be the largest size of the set $V'$ takes in \Cref{alg:kcliqueListingCommunity}. Then, the cost of \Cref{alg:kcliqueListingCommunity}, excluding the cost of preprocessing, is $\bigo{k m \left(\frac{\gamma + 4 - k}{2}\right)^{k-2}}$ work and $O(k\log(\gamma+1))$ depth to list all $k$-cliques.
\end{theorem}
The analysis is basically the same as in \Cref{sec:work-outer-loop}, except that we have a better bound for the induced subgraphs given by the maximum size of $|V'|=\gamma$. %

We again propose three different variants on implementing the preprocessing in \Cref{alg:kcliqueListingCommunity}. The simplest way is the greedy approach outlined above, which leads to the best work.

The problem with this approach is that it has linear depth in $n$ to compute the edge order. Instead of computing the order of the edges greedily, we can compute an approximation of the order, summarized in \Cref{alg:apxCommDegOrder} and described in the following:

Compute the triangles $T$ in $G$. For each edge, keep track of the triangles that contain it (and vice versa). Moreover, for each edge $e$, keep a count of the remaining number of triangles $\cardinality{C(e)}$ that contain $e$. Then, repeat the following until the graph has no edges left: Select all the edges with $\cardinality{C(e)}\leq (3+\epsilon)T/m$ triangles. The constant $\epsilon$ must be positive. Remove those edges from the graph and add them to the total order, breaking ties arbitrarily. For each triangle that contained an edge in the removed set of edges, its edges $e$ need to update their count $\cardinality{C(e)}$. Then, update the overall count of remaining triangles $T$.

\begin{algorithm}
	\caption{Computing a $(3+\epsilon)$-approximate community degeneracy order of the graph $G$}\label{alg:apxCommDegOrder}
	\begin{algorithmic}[1]
		\State Compute the triangles $T$ in $G$ and link them with their edges.
		\State Count for each edge the triangles $\cardinality{C(e)}$ that contain it.
		\While{$E\neq \emptyset$}\label{line:loop:kclique}
			\State $E' \gets $ the edges with at most $\cardinality{C(e)}\leq (3+\epsilon)T/m$ triangles.
			\State Add the edges in $E'$ to the total order, break ties arbitrarily.
			\State Remove $E'$ from the graph
			\State Update each $C(e)$ to reflect the removal of $E'$.
			\State Update the number of triangles $T$.
		\EndWhile
	\end{algorithmic}
\end{algorithm}

Next, we make a few observations to help us bounding the work, depth, and quality of the approximate community degeneracy order.

\begin{observation}\label{lem:commdeg-triangles}
A graph with community degeneracy $\sigma$ has at most $\sigma m$ triangles.
\end{observation}
\begin{proof}
Remove the edges of the graph in a greedy manner, picking the edge with the smallest remaining number of triangles that contain it next. By definition of the community degeneracy, each subgraph encountered this way has an edge with $\sigma$ triangles that contain it. The process terminates after $m$ steps. Hence, the process removes at most $m\sigma$ triangles until no triangle is left.
\end{proof}

\begin{observation}\label{obs:apxCommDegOrder-iteration}
	\Cref{alg:apxCommDegOrder} terminates after $\bigo{ \log_{1+\epsilon} m}$ iterations
\end{observation}
\begin{proof}
From \Cref{lem:commdeg-triangles}, it follows that the average number of triangles per edge is at most $3\sigma T/m$ in each iteration (each triangle is counted once for each of its edges). From any set of numbers, at most a $\frac{1}{1+\epsilon}$ fraction of numbers are larger than $(1+\epsilon)$ times the average. Hence, each iteration reduces the number of edges by a factor at least $(1+\epsilon)$.
\end{proof}

\begin{lemma}\label{lem:apxCommDegOrder-WD}
	\Cref{alg:apxCommDegOrder} computes an order $\preceq$ on the edges such that for each edge $e$, the set $C_{(V, E[e\preceq])} (e)$ has size at most $(3+\epsilon)\sigma$. It takes $\bigo{ms + m\sigma}$ work and $\bigo{\log n \log_{1+\epsilon} n }$ depth. 
\end{lemma}
\begin{proof}
Listing the triangles takes $\bigo{ms}$ work and $\bigo{\log^2 n}$ depth~\cite{Chiba1985}. 
Each iteration of the loop takes work proportional to the number of removed edges and $\bigo{\log n}$ depth.
By \Cref{obs:apxCommDegOrder-iteration}, the algorithm terminates after $\bigo{\log_{1+\epsilon} n}$ iterations.

For the bound on the size of the communities, note that by \Cref{lem:commdeg-triangles}, $T/m \leq \sigma$. Hence, each of the removed edges satisfy $\cardinality{C(e)}\leq (3+\epsilon) T/ m \leq (3+\epsilon)\sigma$.
\end{proof}
The results for the community degeneracy in \Cref{tab:bounds} follow from \Cref{thm:comm-deg-algo-work} by using either the sequential greedy preprocessing, \Cref{lem:apxCommDegOrder-WD}, or by using a hybrid approach similar to \Cref{sec:order-hybrid}.

\section{Conclusion and Future Work}
We presented work improvements over previous algorithms on $k$-clique listing in $s$-degenerate and $\sigma$-community-degenerate graphs. The improvements are exponential in the clique size $k$, in particular when $k=\Theta(s)$.

Many interesting open questions remain. There remains a gap of $\bigo{ns/m}$ between the work of our best poly-logarithmic depth algorithm and our lowest work algorithm. Can we do better? Moreover, it would be interesting if the work can be further reduced when $k$ is not constant. It might be interesting to consider generalizations that extend the cliques by larger motifs such as triangles. Finally, are there other interesting classes of graphs where the clique problem is tractable? 

\bibliographystyle{plain}
\bibliography{../misc_notes/refs2,refs}

\appendix

\section{Work Bounds Continued}\label{sec:work-bounds-remaining}

For completeness, we include the remaining derivation of the work of our algorithm. We prove the contribution of the intersections and the edge existence probing separately in \Cref{sec:work-recursive-bounds} and combine them with the clique listing term from \Cref{sec:work-analysis} in \Cref{sec:work-outer-loop}.

\subsection{Recursive Work Cost}\label{sec:work-recursive-bounds}
We bound the work incurred by \Cref{alg:recursiveListing} for intersecting the edge neighborhoods. Note that because every community (and candidate set $I$) has size at most $\gamma$ by assumption and the sets are sorted in advance, the cost of any particular intersection is at most $d \gamma$ for some constant $d$.
\begin{lemma}
Let $Q(c, I)$ be the work of \Cref{alg:recursiveListing} incurred for intersection of sets, for a graph where the largest community has size $\gamma$. Then, we have for $c\ge 3$,%
$$
Q(c, I) \le d\gamma \enspace \left(\sum_{p=0}^{r(c-3)}\left(\frac{\cardinality{I}-c+2}{2}\right)^{2p}\right)\enspace \cardinality{\relevantedge{G[I]}{c-2}}. \label{eq:work:intersection}
$$
\end{lemma}
\begin{proof}
The proof is again by induction on $c$. For the base case $c=2$ there are trivially no intersections required. For the cases $c=3$ and $c=4$, the algorithm only needs to intersect the community of each edge relevant with respect to $1$ (resp., $2$) once with $I$ and the bound holds. %
Assume that the bound holds for all $c'\leq c$. We show that it also holds in the case $c+2$. For that, we start with a recursive expression for the work and apply the induction hypothesis. %
	\begin{align}
		&\enspace Q(c, I) \\
		\le &\sum_{e\in\relevantedge{G[I]}{c-2}} d\gamma + Q(c-2, I \cap \community{e}), \\
		\le &\sum_{e\in\relevantedge{G[I]}{c-2}} d \gamma \enspace + \notag \\
		&d \gamma\enspace \left(\sum_{p=0}^{r(c-5)}\left(\frac{\cardinality{I \cap \community{e}}-c+4}{2}\right)^{2p}\right) 
		\enspace \cardinality{\relevantedge{G[I \cap \community{e}]}{c-4}}
	\end{align}
	We rearrange the terms and use that $\cardinality{I \cap \community{e}}\le \cardinality{I}-2$:
	\begin{align}
		\le \enspace &d \gamma \enspace \cardinality{\relevantedge{G[I]}{c-2}} \enspace + \notag\\
		& d \gamma \enspace \left(\sum_{p=0}^{r(c-5)}\left(\frac{\cardinality{I}-c+2}{2}\right)^{2p}\right) \sum_{e\in\relevantedge{G[I]}{c-2}}
		\cardinality{\relevantedge{G[I \cap \community{e}]}{c-4}} \enspace .
	\end{align}
	Now, we apply \Cref{lem:relevantEdgeRecursion}:
	\begin{align}
		\le \enspace &d \gamma \enspace \cardinality{\relevantedge{G[I]}{c-2}} \enspace + \notag \\
		& d \gamma \enspace \left(\sum_{p=0}^{r(c-5)}\left(\frac{\cardinality{I}-c+2}{2}\right)^{2p}\right)
		\enspace  \left(\frac{\cardinality{I}-c+2}{2}\right)^2 \cardinality{\relevantedge{G[I]}{c-2}} \\
		\le \enspace &d \gamma \enspace \cardinality{\relevantedge{G[I]}{c-2}} \enspace + \notag \\
		& d \gamma \enspace \left(\sum_{p=1}^{r(c-5)+1}\left(\frac{\cardinality{I}-c+2}{2}\right)^{2p}\right)
		\cdot \cardinality{\relevantedge{G[I]}{c-2}} \enspace.
	\end{align}
	Now, we use \Cref{obs:level-function} and use pull the first term into the sum as well. This results in the desired form, which concludes the proof:
	\begin{align}
		Q(c,I) &\le d \gamma\cdot\left(\sum_{p=0}^{r(c-3)}\left(\frac{\cardinality{I}-c+2}{2}\right)^{2p}\right)
		\enspace \cardinality{\relevantedge{G[I]}{c-2}} \enspace.
	\end{align}
\end{proof}

Next, we bound the cost incurred for probing for edge existence.
By building a perfect hash table (or a adjacency matrix for each subgraph induced by $C(e)$ for every edges $e$), the cost to probe for the existence of an edge is at most some constant $f$. 
\begin{lemma}
Let $S(c, I)$ be the work of \Cref{alg:recursiveListing} incurred for probing for edge existence. Then, we have for $c\ge2$,
$$S(c, I) \le f\cdot\left( \sum_{p=0}^{r(c-2)}\left( \frac{\cardinality{I}-c+2}{2} \right)^{2p} \right)\cardinality{\relevantpair{I}{c-2}} \label{eq:work:testing}$$
\end{lemma}
\begin{proof}
	The proof is by induction. %
	If $c=2$ or $c=3$, then the algorithm only needs to test all pairs relevant with respect to $0$ or $1$ in $I$ once and the bound holds.
Assume that the bound holds for all $c'\leq c$. We show that it also holds in the case $c+2$. We write:
	\begin{align}
		\enspace &S(c, I) \\
		\le  \enspace &f \enspace \cardinality{\relevantpair{I}{c-2}}+ \sum_{e\in\relevantedge{G[I]}{c-2}}S(c-2, I  \cap \community{e}) \\
		\le \enspace &f \enspace \cardinality{\relevantpair{I}{c-2}}\enspace + \notag \\
		& \sum_{e\in\relevantedge{G[I]}{c-2}}\left( \sum_{p=0}^{r(c-4)}\left(
		\frac{\cardinality{I  \cap \community{e}}-c+4}{2}\right)^{2p} \right)\cardinality{\relevantpair{I  \cap \community{e}}{c-4}}
	\end{align}
	Use that $\cardinality{I  \cap \community{e}}\le \cardinality{I}-2$, to pull two factors in front of the sum:
	\begin{align}
		\le \enspace &f \enspace \cardinality{\relevantpair{I}{c-2}}  \enspace  + \\
		& f\enspace\left(\sum_{p=0}^{r(c-4)}\left(\frac{\cardinality{I}-c+2}{2}\right)^{2p}\right)\enspace
		\sum_{e\in\relevantedge{G[I]}{c-2}}\cardinality{\relevantpair{I  \cap \community{e}}{c-4}} \enspace .
	\end{align}
	We use \Cref{lem:relevantEdgeRecursion} on the graph $G[I]$ and use that every relevant edge is also a relevant pair:
	\begin{align}
		\le \enspace &f\enspace \cardinality{\relevantpair{I}{c-2}} \enspace + \notag \\
		&f\enspace\left(\sum_{p=0}^{r(c-4)}\left(\frac{\cardinality{I}-c+2}{2}\right)^{2p}\right)\enspace
		\left(\frac{\cardinality{I}-c+2}{2}\right)^2\cardinality{\relevantpair{I}{c-2}} \\
		= \enspace &f \enspace \cardinality{\relevantpair{I}{c-2}} \enspace +  
		f\enspace\left(\sum_{p=1}^{r(c-4)+1}\left(\frac{\cardinality{I}-c+2}{2}\right)^{2p}\right)\enspace
		\cardinality{\relevantpair{I}{c-2}}
	\end{align}
	Now, we apply \cref{obs:level-function} and pull the first term into the sum. We arrive at the desired form, concluding the proof:
	\begin{align}
		S(c,I) \le f\enspace\left(\sum_{p=0}^{r(c-2)}\left(\frac{\cardinality{I}-c+2}{2}\right)^{2p}\right)\enspace
		\cardinality{\relevantpair{I}{c-2}}
	\end{align}
\end{proof}

\subsection{Overall Work Bound}\label{sec:work-outer-loop}
So far, we derived the work of \Cref{alg:recursiveListing}. Now, we will apply the outer most loop in \Cref{alg:kcliqueListing}: The loop itself simply iterates over all relevant edges and calls \Cref{alg:recursiveListing} with $c=k-2$. Note, that we can simplify the call to the base case function $b$ with \Cref{obs:level-function} and the fact that $k-2$ and $k$ are either both even or both odd.
\begin{align}
	W &\le \sum_{e\in\relevantedge{G}{k-2}}W(k-2, \community{e}), \\
	&\le f \sum_{e\in\relevantedge{G}{k-2}} \left( \sum_{p=0}^{r(k-4)}\left( \frac{\cardinality{\community{e}}-k+4}{2} \right)^{2p} \right)\cardinality{\relevantpair{\community{e}}{k-4}} \notag \\
	&+  d\gamma \sum_{e\in\relevantedge{G}{k-2}}\left(\sum_{p=0}^{r(k-5)}\left(\frac{\cardinality{\community{e}}-k+4}{2}\right)^{2p}\right)
	 \cardinality{\relevantedge{G[\community{e}]}{k-4}}  \notag \\
	&+ \sum_{e\in\relevantedge{G}{k-2}} \left(\frac{\cardinality{\community{e}}-k+4}{2}\right)^{2r(k-4)}\cardinality{\relevantedge{G[\community{e}]}{k-4}} \notag \\
	& \enspace \enspace \enspace \enspace \cdot b(k, \cardinality{\community{e}}-2r(k-2)) \enspace . 
\end{align}
Use that $\cardinality{\community{e}} \le \gamma$ by assumption:
\begin{align}
	W &\le f \enspace \left( \sum_{p=0}^{r(k-4)}\left( \frac{\gamma-k+4}{2} \right)^{2p} \right) \sum_{e\in\relevantedge{G}{k-2}}\cardinality{\relevantpair{\community{e}}{k-4}} \notag \\
	&+ d\gamma \enspace \left(\sum_{p=0}^{r(k-5)}\left(\frac{\gamma-k+4}{2}\right)^{2p}\right)
	 \sum_{e\in\relevantedge{G}{k-2}} \cardinality{\relevantedge{G[\community{e}]}{k-4}} \notag \\
	&+ \left(\frac{\gamma-k+4}{2}\right)^{2r(k-4)} b(k, \gamma-2r(k-2)) \sum_{e\in\relevantedge{G}{k-2}}\cardinality{\relevantedge{G[\community{e}]}{k-4}} .
\end{align}
Now, apply \Cref{lem:upperBoundSumRelevantInducedEdgesCommunityDegeneracy} to bound the number of relevant edges and pairs in the subgraphs induced by all relevant edges in $G$:
\begin{align}
	W &\le 2f \cdot\left( \sum_{p=1}^{r(k-4)+1}\left( \frac{\gamma-k+4}{2} \right)^{2p} \right)\enspace \cardinality{\relevantedge{G}{k-2}} \notag \\
	&+ 2d\gamma\cdot\left(\sum_{p=1}^{r(k-5)+1}\left(\frac{\gamma-k+4}{2}\right)^{2p}\right)
	\enspace \cardinality{\relevantedge{G}{k-2}} \notag \\
	&+ 2\left(\frac{\gamma-k+4}{2}\right)^{2r(k-4)+2} \enspace b(k, \gamma-2r(k)) \enspace \cardinality{\relevantedge{G}{k-2}}. 
\end{align}
\begin{align}
	&= 2f \cdot\left( \sum_{p=1}^{r(k-2)}\left( \frac{\gamma-k+4}{2} \right)^{2p} \right) \enspace \cardinality{\relevantedge{G}{k-2}} \notag \\
	&+ 2d\gamma\cdot\left(\sum_{p=1}^{r(k-3)}\left(\frac{\gamma-k+4}{2}\right)^{2p}\right)
	\enspace \cardinality{\relevantedge{G}{k-2}} \notag \\
	&+ 2\left(\frac{\gamma-k+4}{2}\right)^{2r(k-2)} b(k, \gamma-2r(k-2)) \enspace \cardinality{\relevantedge{G}{k-2}} \notag \\
	&\le 2f \cdot r(k-2)\left( \frac{\gamma-k+4}{2} \right)^{2r(k-2)} \cardinality{\relevantedge{G}{k-2}} \notag \\
	&+ 2\gamma\cdot r(k-2)\left(\frac{\gamma-k+4}{2}\right)^{2r(k-3)} \notag
	\enspace \cardinality{\relevantedge{G}{k-2}}\\
	&+ 2\left(\frac{\gamma-k+4}{2}\right)^{2r(k-2)} \cdot b(k, \gamma-2r(k-2)) \enspace \cardinality{\relevantedge{G}{k-2}}.
\end{align}
To further simplify the above expression, we make use of a simple observation:
\begin{observation}\label{obs:simple-arithmetic-bound}
	For $k\ge 1, k\le \gamma+2$, it holds that
	\begin{displaymath}
		2\gamma \le 2k(\gamma+4-k).
	\end{displaymath}
\end{observation}
Now recall the definition of the level function:
\begin{displaymath}
	r(k)  := \left\lfloor\frac{k}{2}\right\rfloor.
\end{displaymath}
Due to the floor operator, it evaluates to the same value for $k$ even and $k'=k+1$ odd, but the base case incurs different costs for $k$ even and $k'$ odd. Let's take a look at the listing part of the cost:
\begin{align}
	W(k\ \text{even}) &\le 2\left(\frac{\gamma-k+4}{2}\right)^{k-2} \cdot k, \\
	W(k'\ \text{odd}) &\le 2\left(\frac{\gamma-k'+4}{2}\right)^{k'-3} \cdot k'(\gamma - (k'-3)), \\
	&\le 2\left(\frac{\gamma-k'+4}{2}\right)^{k'-3} \cdot k'2\frac{\gamma - k' + 4}{2}, \\
	&\le 4\left(\frac{\gamma-k'+4}{2}\right)^{k'-2} \cdot k'
\end{align}
Together with \Cref{obs:simple-arithmetic-bound}, we see that if $k$ is even, then the test and listing terms are dominant. In the odd case, the listing term is the dominant one, closely followed by the testing term. This concludes the wo

\section{Experimental Evaluation}

In this section, we present an empirical evaluation of \Cref{alg:kcliqueListing}. For all experiments, we use the (exact) degeneracy order from \Cref{sec:orientation}. We compare our implementation with the state-of-the-art implementations by Danisch et al.~\cite{Danisch2018} and Shi et al.~\cite{Shi2020}.


\subsection{Data Set}
We assembled a set of graphs to test the different implementations. These graphs represent different fields of applications and exhibit different structural properties~\cite{Besta2020}. \Cref{tab:graphStats} provides a brief overview of the graphs. The graphs `Tech-as-skitter', `Ca-DBLP-2012' and `Orkut' have already been used previously for benchmarking by Danisch et al.~\cite{Danisch2018} and Shi et al.~\cite{Shi2020}.  Orkut is a relatively large online social network, Ca-DBLP-2012 is the collaboration network of the DBLP 2012, and Tech-as-Skitter is an internet topology graph, generated in 2005 by daily traceroutes~\cite{Snapnets}.

In addition, we added some non-standard graphs to the set: `Gearbox', `Chebyshev4', `Bio-SC-HT' and `Jester2'. Gearbox is a structural network from an aircraft flap actuator. Chebyshev4 is a structural network as well. It is derived from a 4th order Chebyshev scheme for numerical integration. Jester2 is a joke-rating network, and Bio-SC-HT is a network derived from functional gene associations~\cite{NetworkRepository}. 


\begin{table}
	\resizebox{\linewidth}{!}{
	\begin{tabular}{lrrrrrrr}
		\toprule
		Graph & $\cardinality{V}$ & $\cardinality{E}$ & $\cardinality{T}$ & $s$ & $\tfrac{\cardinality{E}}{\cardinality{V}}$ & $\tfrac{\cardinality{T}}{\cardinality{V}}$ & $\tfrac{\cardinality{T}}{\cardinality{E}}$ \\ \midrule
		Orkut & $3.1$M & $117.2$M & $627.6$M & $253$ & $38.1$ & $204.6$ & $5.4$ \\
		Ca-DBLP-2012 & $317$K & $1$M & $2.2$M & $113$ & $3.3$ & $7$ & $2.1$ \\
		Tech-As-Skitter & $1.7$M & $11.1$M & $28.8$M & $111$ & $6.5$ & $17$ & $2.6$ \\
		Gearbox & $153.7$K & $4.5$M & $4.6$M & $44$ & $29$ & $30$ & $1$ \\
		Chebyshev4 & $68$K & $1.9$M & $28.9$M & $68$ & $28.9$ &  $424.2$ & $14.7$ \\
		Jester2 & $50.1$K & $1.7$M & $35.6$M & $128$ & $34.1$ &  $703.3$ & $20.6$ \\
		Bio-SC-HT & $2084$ & $63$K & $1.4$M & $100$ & $30.2$ & $670.7$ & $22.2$ \\
	\end{tabular}
}
\vspace{1em}
	\caption{Overview over the selected graphs. $s$ is the degeneracy; $T$ denotes the triangles in the graph. 
	All graphs are publicly available~\cite{Snapnets,NetworkRepository} and have been symmetrized.}\label{tab:graphStats}
\end{table}

\subsection{Experimental Setup}
The bulk of our experiments was performed on the Piz Daint Computing platform by CSCS. 
We used the XC40 Multicore Compute nodes, equipped with two Intel\textsuperscript{\textregistered} Xeon\textsuperscript{\textregistered} E5-2695 v4 \@ 2.10GHz, 18 cores/36 threads, 45MB of Intel\textsuperscript{\textregistered} Smart Cache and each with 64/128 GB RAM~\cite{intelProducts,cscsWeb}.

The code was compiled in both instances with GCC 8.3 with the optimization flags \texttt{-O3} and \texttt{-march=native}. Piz Daint supports AVX2. GCC 8.3 supports OpenMP 4.5~\cite{gnuWeb}, which was chosen for parallel programming.

All measurements were repeated at least 10 times, except for the runs for Orkut and Jester with $k\geq 7$, which were repeated $5$ times due to the long running times. The reported times are arithmetic averages over all measurements. 

\subsection{Results}

\begin{figure}
   	\centering
   	\includegraphics[width=8.2cm]{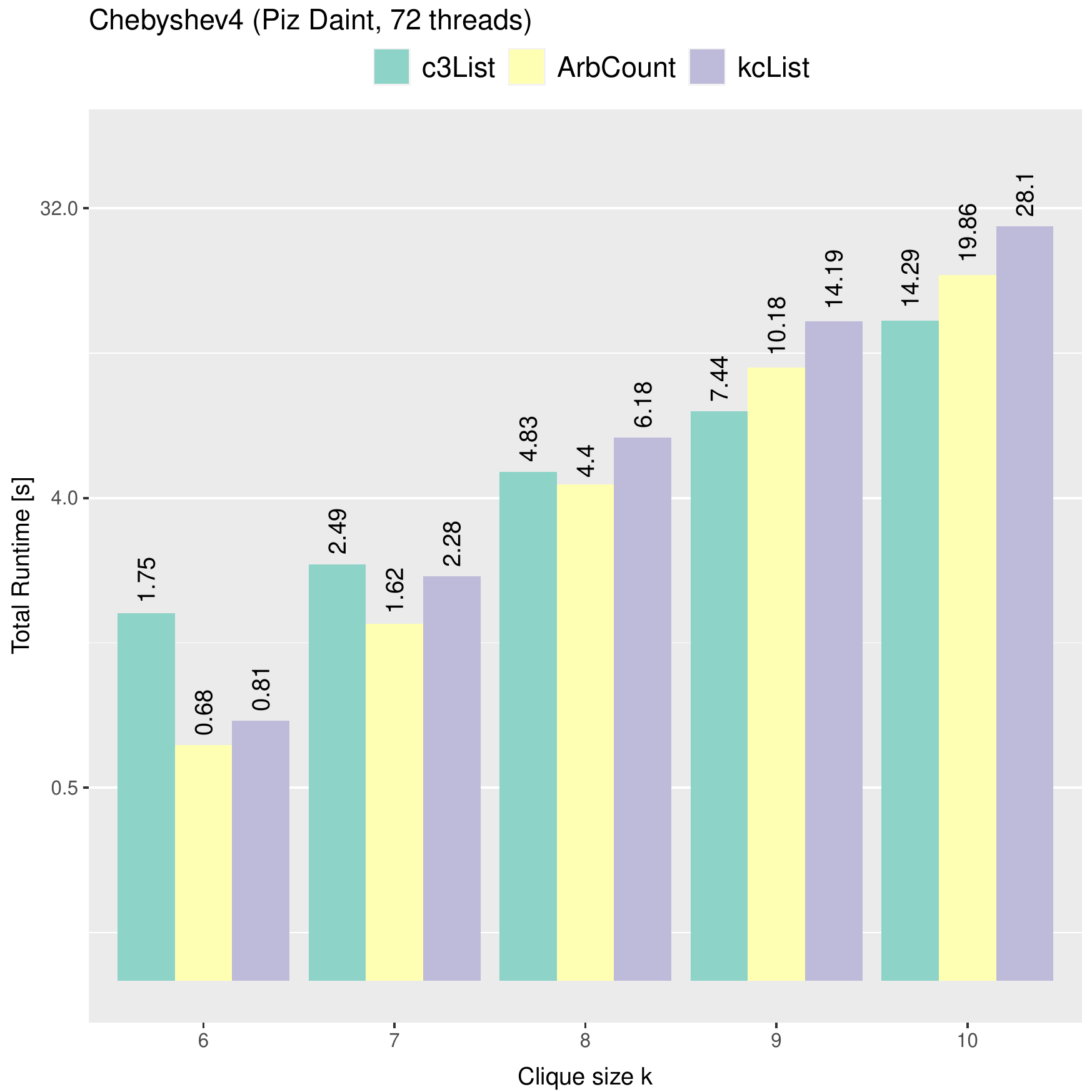}
   	\vspace{1em}
    \caption{Results on the Chebyshev4 graph. }\label{fig:results-chevishev}
\end{figure}

See \Cref{fig:results-chevishev}, \Cref{fig:results1} and \Cref{fig:results2} for the comparison of our algorithm (c3List) with Shi et al. (ArbCount~\cite{Shi2020}) and Danisch et al. (kcList~\cite{Danisch2018}). We report the runtime for $72$ threads and vary the clique size $k$ from $6$ to $10$.

The empirical standard deviation of the runtimes is less than $5.2 \%$ for clique sizes $k\geq 8$ for all algorithms and all graphs except the Gearbox graphs (which can be solved quickly). For Gearbox graphs with $k=10$, the standard deviation is less than $6.4\%$. Other results need to be interpreted with care, as the standard deviation is larger than $10\%$ for at least one graph and algorithm.

For small clique sizes ($k<8$), it depends on the graph which algorithm is fastest. These results also have a higher variance.

For larger clique sizes ($k\ge 8$), Shi et al.~\cite{Shi2020} generally outperform Danisch et al.~\cite{Danisch2018}. Our community-centric algorithm is faster than both on a majority of the instances starting at $k=8$. The advantage of our implementation over the others generally grows with the clique size. This trend might be because of the non-constant base to the exponent in the work bounds of our algorithm.

Overall, for cliques of size $k\geq 9$, our algorithm is slower in one instance (Orkut) and outperforms the others by $3.4-37.9\%$. Our algorithm appears to be relatively better when there are few triangles per vertex. For those graphs (Tech-As-Skitter, Gearbox, CA-DBLP-2012), our algorithm is $16.5-37.9\%$ and $13.8-33.7\%$ faster than the next fastest implementation for cliques of size $9$ and $10$, respectively. The good performance on those graphs could be because our pruning technique is particularly effective when there are fewer triangles.

 

\begin{figure*}[p]
	\begin{subfigure}[t]{0.49\textwidth}
   	\centering
   	\includegraphics[width=8.2cm]{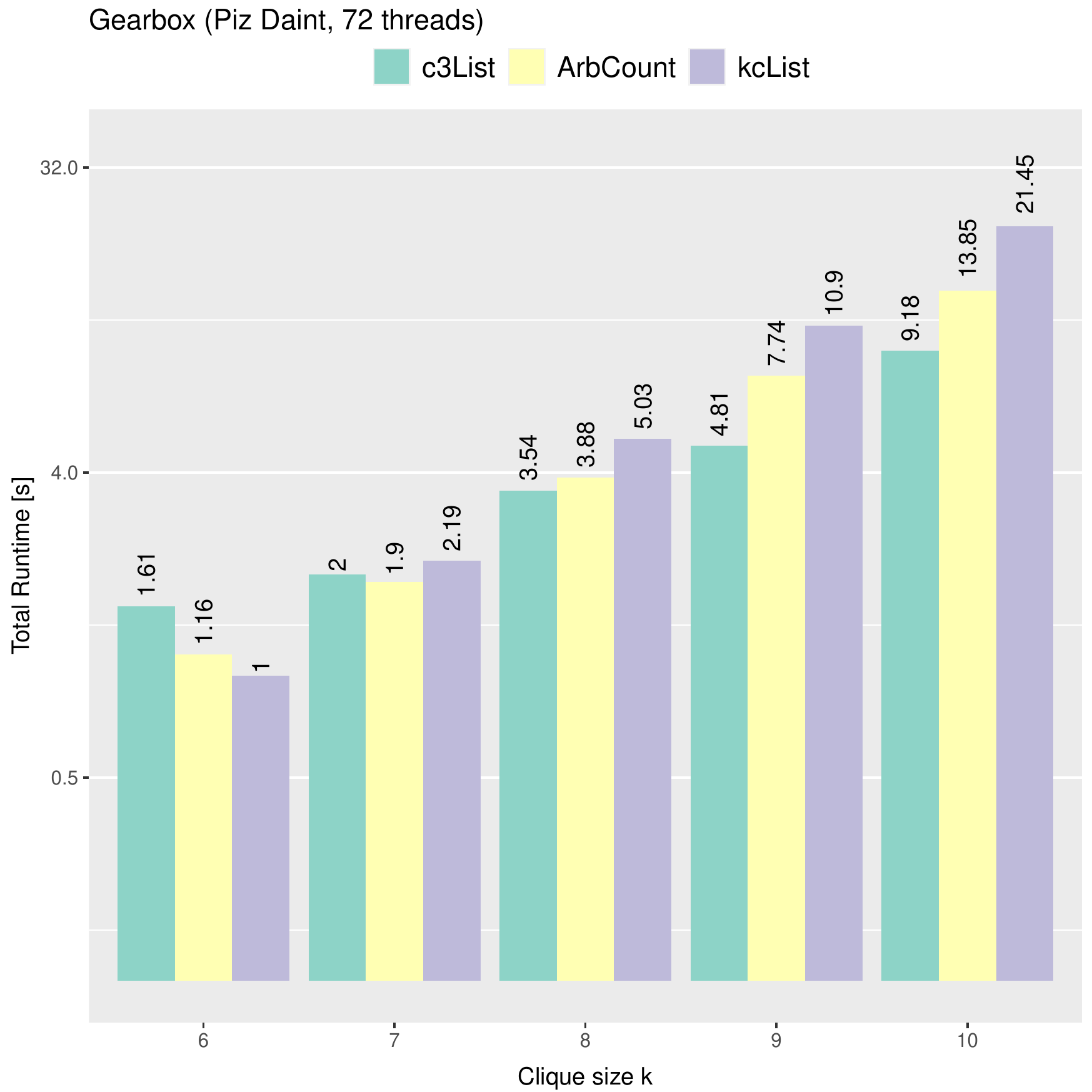}
   	\vspace{1em}
   	\caption{}\label{fig:results-gearbox}
    \vspace{3em}
  \end{subfigure}
    	\begin{subfigure}[t]{0.49\textwidth}
   	\centering
   	\includegraphics[width=8.2cm]{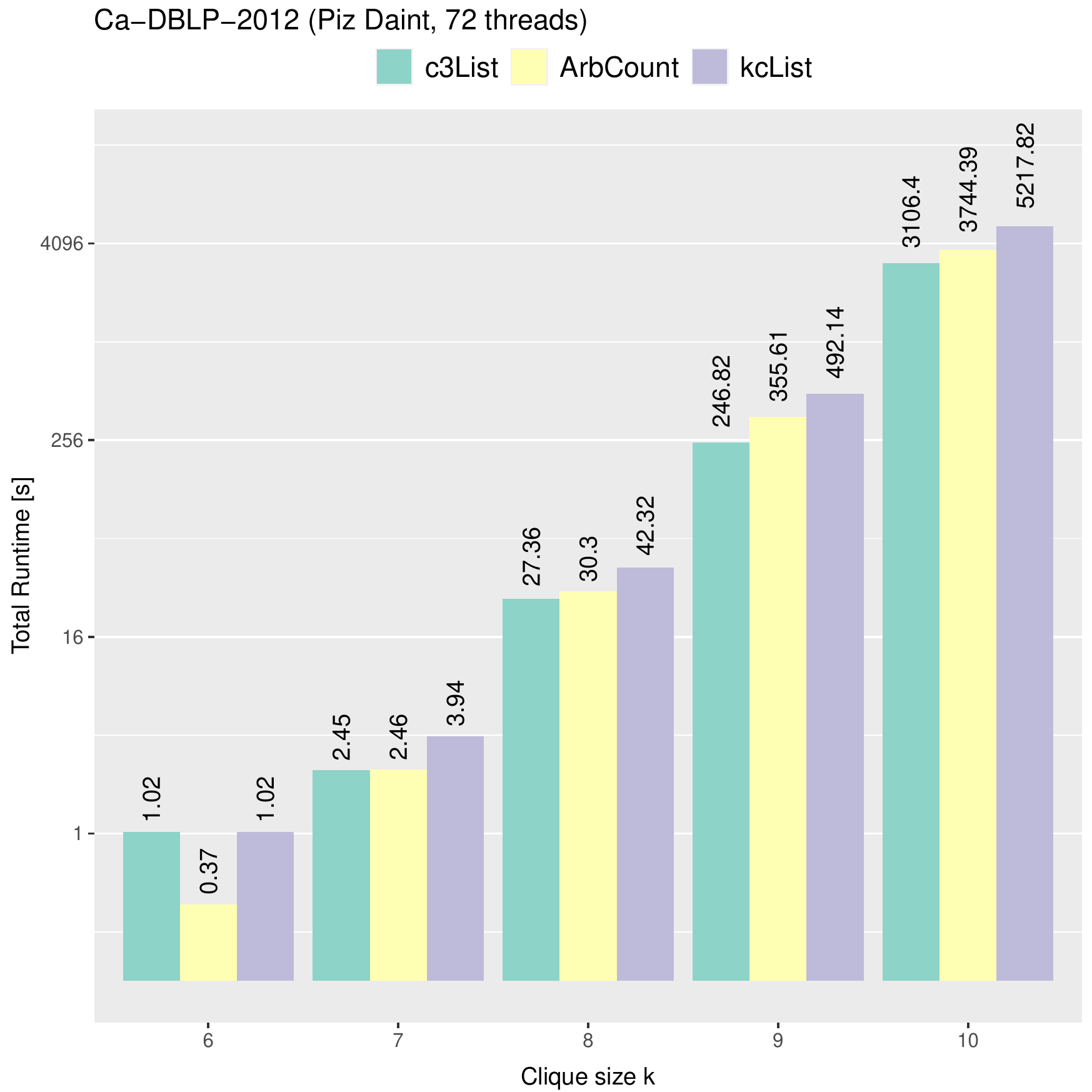}
   	\vspace{1em}
   	\caption{}\label{fig:results-dblp}
   \vspace{2em}
  \end{subfigure}

 \begin{subfigure}[t]{0.49\textwidth}
   	\centering
   	\includegraphics[width=8.2cm]{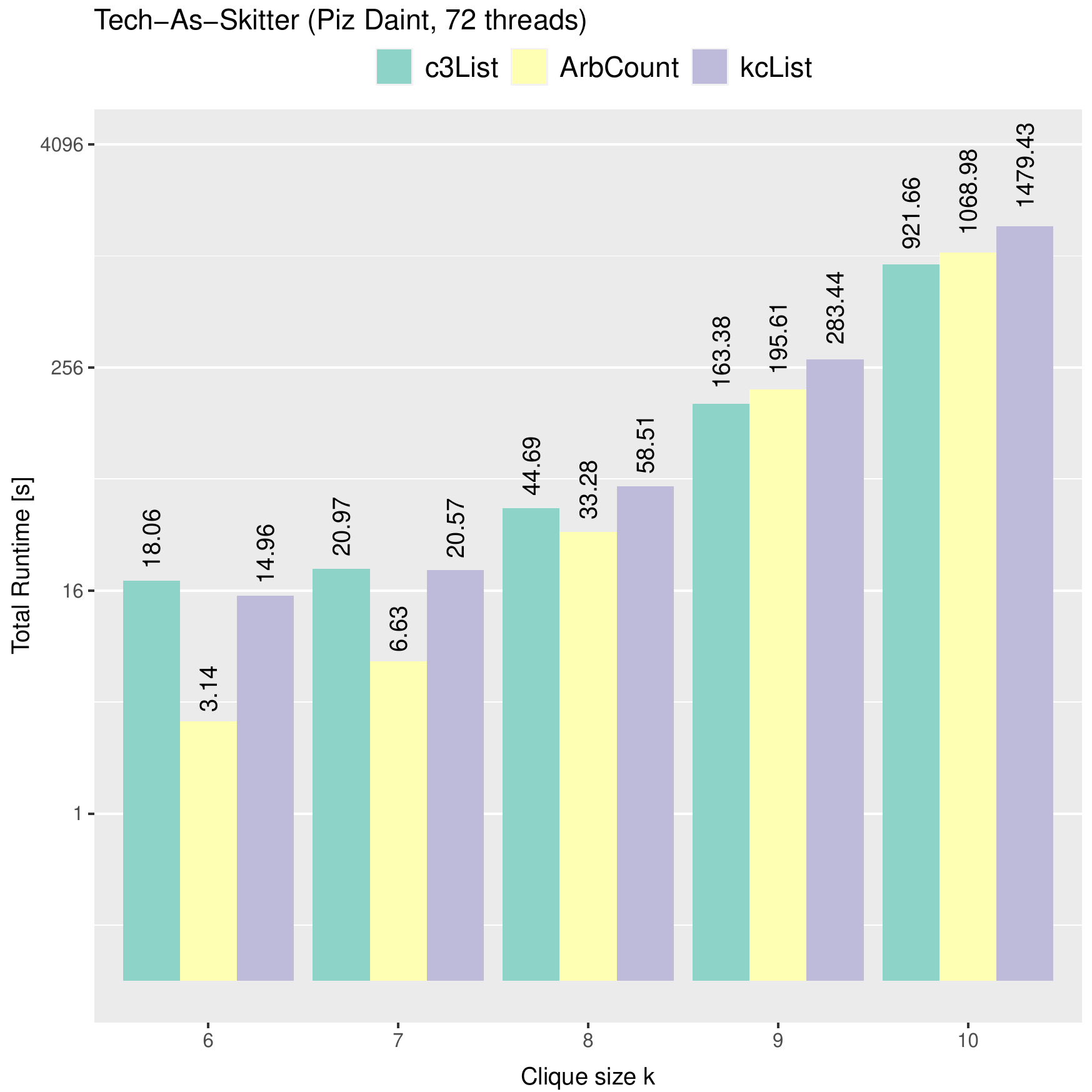}
   	\vspace{1em}
    \caption{}\label{fig:results-skitter}
    \vspace{2em}
  \end{subfigure}
	\begin{subfigure}[t]{0.49\textwidth}
   	\centering
   	\includegraphics[width=8.2cm]{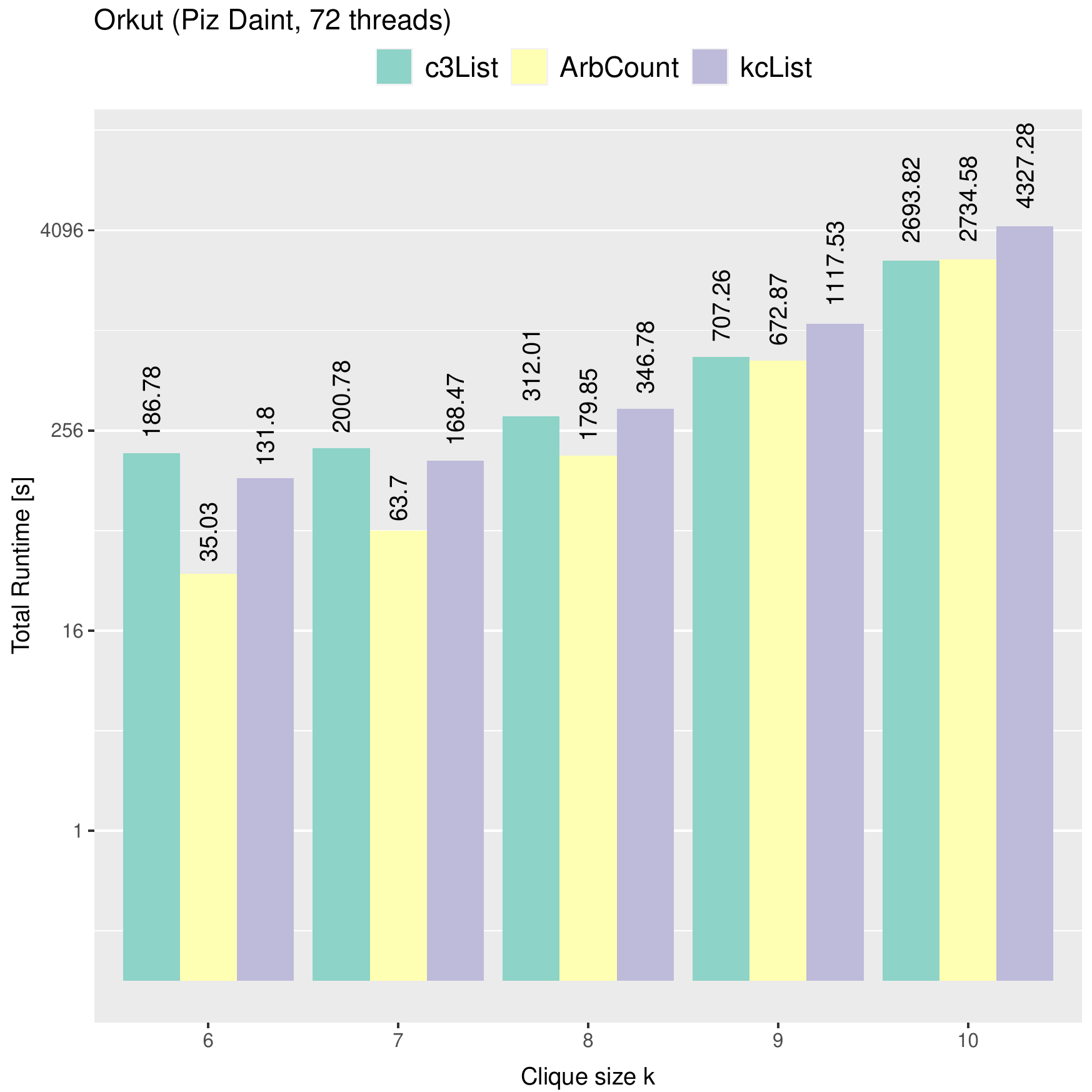}
   	\vspace{1em}
    \caption{}\label{fig:results-orkut}
    \vspace{3em}
  \end{subfigure} 
  \caption{Runtime Results for 72 threads for varying clique sizes. Our algorithm is c3List, ArbCount is by Shi et al.~\cite{Shi2020}, and kcList is by Danisch et al.~\cite{Danisch2018}}\label{fig:results1}
\end{figure*}

\begin{figure*}[p]
	\begin{subfigure}[t]{0.49\textwidth}
   	\centering
   	\includegraphics[width=8.2cm]{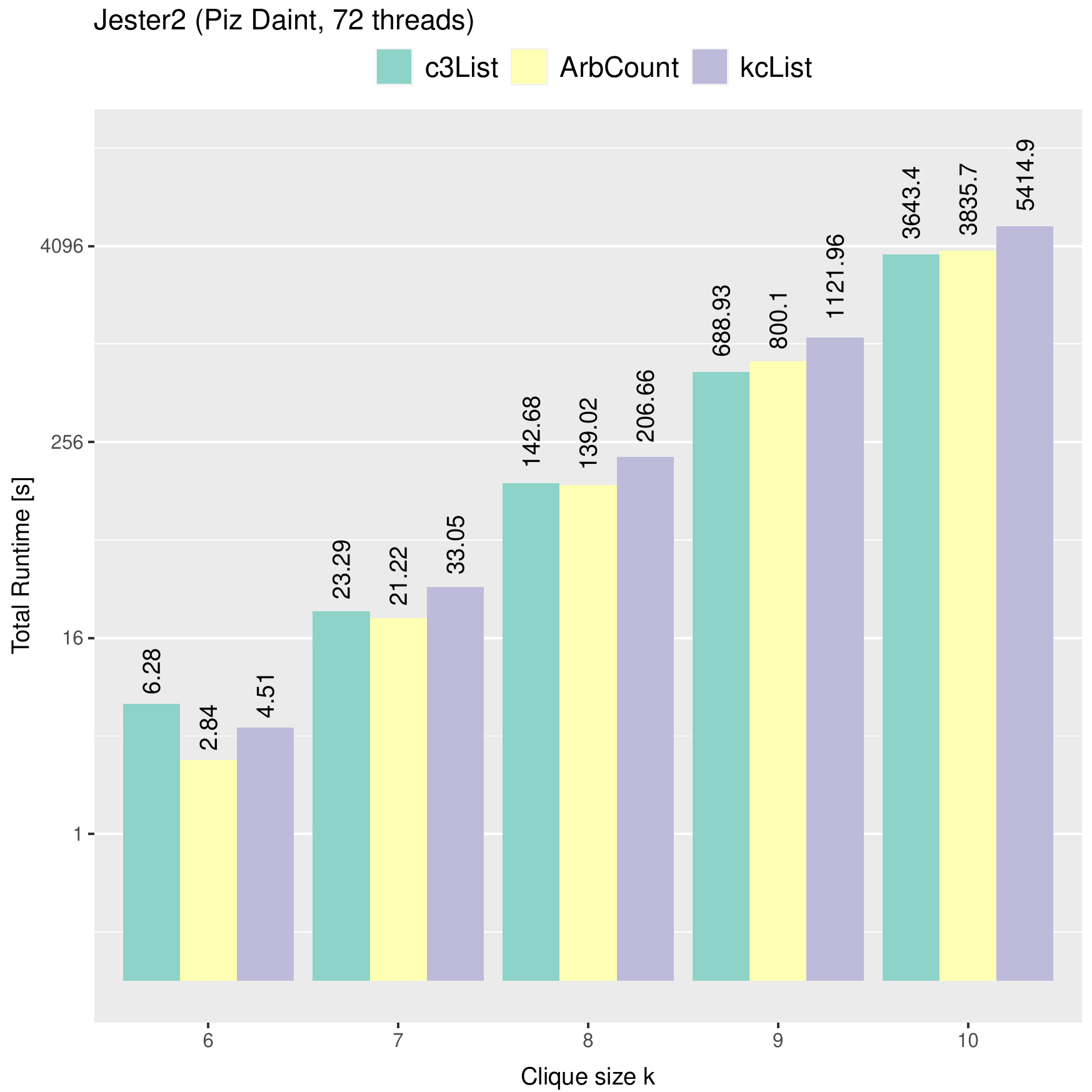}
   	\vspace{1em}
   \caption{}
    \vspace{2em}
  \end{subfigure} 
 \begin{subfigure}[t]{0.49\textwidth}
   	\centering
   	\includegraphics[width=8.2cm]{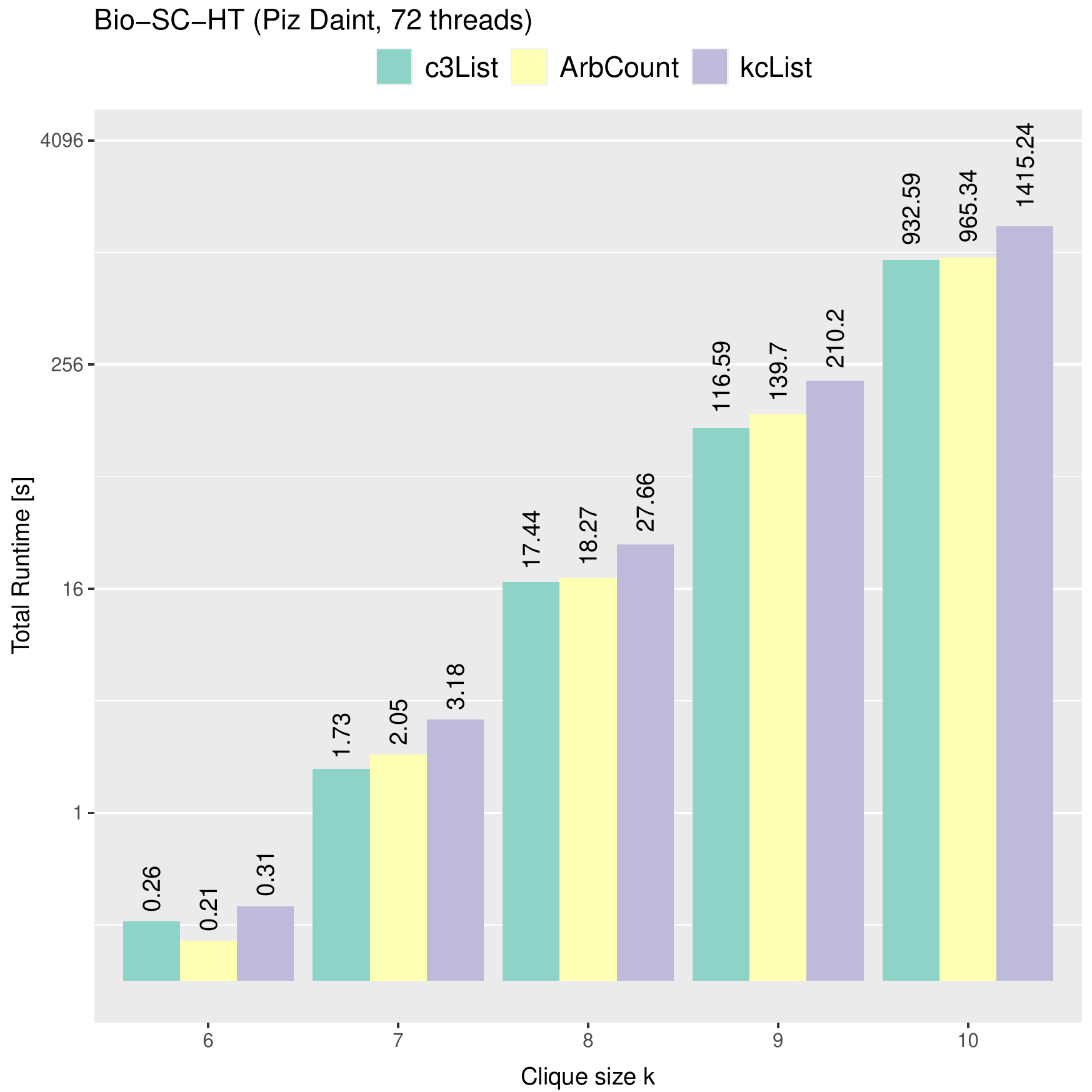}
   	\vspace{1em}
    \caption{}
    \vspace{2em}
  \end{subfigure}
  \caption{Runtime Results for 72 threads for varying clique sizes. Our algorithm is c3List, ArbCount is by Shi et al.~\cite{Shi2020}, and kcList is by Danisch et al.~\cite{Danisch2018}}\label{fig:results2}
\end{figure*}

\end{document}

%% file: extrapackages.tex
\usepackage{datetime}

\usepackage{mathtools}

\usepackage{ifthen}

\usepackage[h]{esvect}

\usepackage{array}

\usepackage{xcolor}

\usepackage{algorithm}
\usepackage{algpseudocode}

\usepackage{tabularx}
\usepackage{booktabs}

\usepackage{cleveref}

\usepackage{xparse}

\usepackage{tikz}

\usepackage{glossaries}
\usepackage{xspace}

\usepackage{ulem}

%% file: macrosetup.tex
\newcommand{\kclique}{$k$-clique\xspace}
\newcommand{\kcliques}{$k$-cliques\xspace}

\newcommand{\ie}{i.e.,\xspace}
\newcommand{\eg}{e.g.,\xspace}

\DeclarePairedDelimiter\abs{\lvert}{\rvert}

\renewcommand{\epsilon}{\ensuremath\varepsilon}

\renewcommand{\phi}{\ensuremath{\varphi}}

\newcommand{\arboricity}[1][default7null]{%
	\ifthenelse{\equal{#1}{default7null}}%
	{\alpha}%
	{\alpha(#1)}%
}
\newcommand{\corenumber}[1][default7null]{%
	\ifthenelse{\equal{#1}{default7null}}%
	{\mathit{c}}%
	{\mathit{c}(#1)}}

\newcommand{\induce}[1]{[#1]}

\newcommand{\neighbors}[2][default7null]{%
	\ifthenelse{\equal{#1}{default7null}}%
	{N({#2})}%
	{N\induce{#1}({#2})}%
}

\newcommand{\community}[2][default7null]{
	\ifthenelse{\equal{#1}{default7null}}%
	{C({#2})}%
	{C_{#1}(#2)}%
}

\newcommand{\outneighbor}[2]{N_{#2}^{+}({#1})}

\newcommand{\inneighbor}[2]{N_{#2}^{-}({#1})}

\newcommand{\outneighbors}[2][default7null]{%
	\ifthenelse{\equal{#1}{default7null}}%
	{N^+({#2})}%
	{N\induce{#1}^+({#2})}%
}

\newcommand{\inneighbors}[2][default7null]{%
	\ifthenelse{\equal{#1}{default7null}}%
	{N^-({#2})}%
	{N\induce{#1}^-({#2})}%
}

\newcommand{\edgeset}[1][default7null]{%
	\ifthenelse{\equal{#1}{default7null}}%
	{E}%
	{E\induce{#1}}%
}

\newcommand{\diredgeset}[1][default7null]{%
	\ifthenelse{\equal{#1}{default7null}}%
	{\vec{E}}%
	{\vec{E}\induce{#1}}%
}

\newcommand{\nodeset}[1][default7null]{%
	\ifthenelse{\equal{#1}{default7null}}%
	{V}%
	{V\induce{#1}}%
}

\newcommand{\cardinality}[1]{%
	\abs{{#1}}%
}

\newcommand{\relevantout}[2]{%
	\mathbb{P}^+_{#2}\left({#1}\right)%
}

\newcommand{\relevantin}[2]{%
	\mathbb{P}^-_{#2}\left({#1}\right)%
}

\newcommand{\relevantpair}[2]{%
	\mathcal{R}^{P}_{#2}\left({#1}\right)%
}

\newcommand{\relevantedge}[2]{%
	\mathcal{R}^E_{#2}\left({#1}\right)%
}

\newcommand{\relevantedgeout}[2]{%
	\mathbb{E}^+_{#2}\left({#1}\right)%
}

\newcommand{\relevantedgein}[2]{
	\mathbb{E}^-_{#2}\left({#1}\right)%
}

\newcommand{\bigo}[1]{\mathcal{O}\left(#1\right)}

\algblock[PARFOR]{ParFor}{EndFor}
\algblockdefx[PARFOR]{ParFor}{EndFor}[1]{\textbf{for} {#1} \textbf{do in parallel}}{\textbf{end for}}
\algblock[PARFORALL]{ParForAll}{EndFor}
\algblockdefx[PARFORALL]{ParForAll}{EndFor}[1]{\textbf{for all} {#1} \textbf{do in parallel}}{\textbf{end for}}

\definecolor{tcolor}{RGB}{255,255,255}

\definecolor{inversetcolor}{RGB}{0,0,0}
\definecolor{cdefault}{rgb}{0.45,0.45,0.45}

\definecolor{hsbcurrent}{RGB}{158, 50, 53}

\definecolor{hsbold}{RGB}{150, 150, 150}
\definecolor{hsbcandidate}{RGB}{20, 20, 20}
\definecolor{hsbecandidate}{RGB}{20, 20, 20}

\definecolor{highlight1}{RGB}{220, 100, 10}
\definecolor{highlight2}{RGB}{10, 10, 100}

\definecolor{highlightcolor}{RGB}{0, 0, 0}

\colorlet{ccurrent}[rgb]{hsbcurrent}
\colorlet{cold}[rgb]{hsbold}
\colorlet{ccandidate}[rgb]{hsbcandidate}

\colorlet{chigh1}[rgb]{highlight1}
\colorlet{chigh2}[rgb]{highlight2}

\tikzset{vertex/.style={circle,fill=cdefault,minimum size=20pt,inner sep=0pt,text=tcolor,font=\large}}
\tikzset{current/.style={circle,fill=ccurrent, minimum size=20pt,inner sep=0pt,text=tcolor,font=\large}}
\tikzset{old/.style={circle,fill=cold, minimum size=20pt,inner sep=0pt,text=tcolor,font=\large}}
\tikzset{candidate/.style={circle,fill=ccandidate, minimum size=20pt,inner sep=0pt,text=tcolor,font=\large}}

\tikzset{highLeft/.style={circle,fill=chigh1, minimum size=20pt,inner sep=0pt,text=tcolor}}
\tikzset{highRight/.style={circle,fill=chigh2, minimum size=20pt,inner sep=0pt,text=tcolor}}

\tikzset{edge/.style={cdefault,thick,dotted}}
\tikzset{ecurrent/.style={ccurrent,thick}}
\tikzset{eold/.style={cold,thick}}
\tikzset{ecandidate/.style={hsbecandidate,thick}}
\tikzset{eoldthk/.style={cold,very thick}}
\tikzset{ecurthk/.style={ccurrent,ultra thick}}

\definecolor{hsbcandidate1}{Hsb}{330, 0.65, 0.55}
\definecolor{hsbcandidate2}{Hsb}{330, 0.65, 0.65}
\definecolor{hsbcandidate3}{Hsb}{330, 0.65, 0.75}
\definecolor{hsbcandidate4}{Hsb}{330, 0.65, 0.85}
\definecolor{hsbcandidate5}{Hsb}{330, 0.65, 0.95}

\colorlet{ccandidate1}[rgb]{highlightcolor}
\colorlet{ccandidate2}[rgb]{highlightcolor}
\colorlet{ccandidate3}[rgb]{highlightcolor}
\colorlet{ccandidate4}[rgb]{highlightcolor}
\colorlet{ccandidate5}[rgb]{highlightcolor}


%% file: acronyms.tex
\makeglossary

\setacronymstyle{long-short}
\newacronym{sct}{SCT}{Succint Clique Tree}
\newacronym{alg}{Alg}{The Algorithm}
\newacronym{dfs}{DFS}{Depth First Search}
\newacronym{bfs}{BFS}{Breadth First Search}
\newacronym{fpm}{FPM}{Frequent Pattern Mining}
\newacronym{dag}{DAG}{Directed Acyclic Graph}

%% file: spaa2021-KClique.bbl
\begin{thebibliography}{10}

\bibitem{aggarwal2010managing}
Charu~C Aggarwal, Haixun Wang, et~al.
\newblock {\em Managing and mining graph data}, volume~40.
\newblock Springer, 2010.

\bibitem{almasri2021k}
Mohammad Almasri, Izzat~El Hajj, Rakesh Nagi, Jinjun Xiong, and Wen-mei Hwu.
\newblock K-clique counting on gpus.
\newblock {\em arXiv preprint arXiv:2104.13209}, 2021.

\bibitem{Besta20HPCColoring}
M.~Besta, A.~Carigiet, K.~Janda, Z.~Vonarburg-Shmaria, L.~Gianinazzi, and
  T.~Hoefler.
\newblock High-performance parallel graph coloring with strong guarantees on
  work, depth and quality.
\newblock In {\em 2020 SC20: International Conference for High Performance
  Computing, Networking, Storage and Analysis (SC)}, pages 1401--1417, Los
  Alamitos, CA, USA, nov 2020. IEEE Computer Society.

\bibitem{besta2020high}
Maciej Besta, Armon Carigiet, Kacper Janda, Zur Vonarburg-Shmaria, Lukas
  Gianinazzi, and Torsten Hoefler.
\newblock High-performance parallel graph coloring with strong guarantees on
  work, depth, and quality.
\newblock In {\em SC20: International Conference for High Performance
  Computing, Networking, Storage and Analysis}, pages 1--17. IEEE, 2020.

\bibitem{besta2019practice}
Maciej Besta, Marc Fischer, Vasiliki Kalavri, Michael Kapralov, and Torsten
  Hoefler.
\newblock Practice of streaming processing of dynamic graphs: Concepts, models,
  and systems.
\newblock {\em arXiv preprint arXiv:1912.12740}, 2019.

\bibitem{besta2015accelerating}
Maciej Besta and Torsten Hoefler.
\newblock Accelerating irregular computations with hardware transactional
  memory and active messages.
\newblock In {\em Proceedings of the 24th International Symposium on
  High-Performance Parallel and Distributed Computing}, pages 161--172, 2015.

\bibitem{besta2018survey}
Maciej Besta and Torsten Hoefler.
\newblock Survey and taxonomy of lossless graph compression and space-efficient
  graph representations.
\newblock {\em arXiv preprint arXiv:1806.01799}, 2018.

\bibitem{besta2021sisa}
Maciej Besta, Raghavendra Kanakagiri, Grzegorz Kwasniewski, Rachata
  Ausavarungnirun, Jakub Ber{\'a}nek, Konstantinos Kanellopoulos, Kacper Janda,
  Zur Vonarburg-Shmaria, Lukas Gianinazzi, Ioana Stefan, et~al.
\newblock Sisa: Set-centric instruction set architecture for graph mining on
  processing-in-memory systems.
\newblock {\em arXiv preprint arXiv:2104.07582}, 2021.

\bibitem{besta2019demystifying}
Maciej Besta, Emanuel Peter, Robert Gerstenberger, Marc Fischer, Micha{\l}
  Podstawski, Claude Barthels, Gustavo Alonso, and Torsten Hoefler.
\newblock Demystifying graph databases: Analysis and taxonomy of data
  organization, system designs, and graph queries.
\newblock {\em arXiv preprint arXiv:1910.09017}, 2019.

\bibitem{besta2017push}
Maciej Besta, Micha{\l} Podstawski, Linus Groner, Edgar Solomonik, and Torsten
  Hoefler.
\newblock To push or to pull: On reducing communication and synchronization in
  graph computations.
\newblock In {\em Proceedings of the 26th International Symposium on
  High-Performance Parallel and Distributed Computing}, pages 93--104, 2017.

\bibitem{besta2018log}
Maciej Besta, Dimitri Stanojevic, Tijana Zivic, Jagpreet Singh, Maurice
  Hoerold, and Torsten Hoefler.
\newblock Log (graph) a near-optimal high-performance graph representation.
\newblock In {\em Proceedings of the 27th International Conference on Parallel
  Architectures and Compilation Techniques}, pages 1--13, 2018.

\bibitem{besta2021graphminesuite}
Maciej Besta, Zur Vonarburg-Shmaria, Yannick Schaffner, Leonardo Schwarz,
  Grzegorz Kwasniewski, Lukas Gianinazzi, Jakub Beranek, Kacper Janda, Tobias
  Holenstein, Sebastian Leisinger, et~al.
\newblock Graphminesuite: Enabling high-performance and programmable graph
  mining algorithms with set algebra.
\newblock {\em VLDB}, 2021.

\bibitem{Besta2020}
Maciej Besta, Zur Vonarburg-Shmaria, Yannick Schaffner, Leonardo Schwarz,
  Grzegorz Kwasniewski, Lukas Gianinazzi, Jakub Beranek, Kacper Janda, Tobias
  Holenstein, Sebastian Leisinger, Peter Tatkowski, Esref Ozdemir, Adrian
  Balla, Marcin Copik, Marek Koneieczny, OnurMutlu, and Torsten Hoefler.
\newblock Graphminesuite: Enabling high-performance and programmable graph
  mining algorithms [bench. \& analysis.
\newblock unpublished, N.D.

\bibitem{besta2019slim}
Maciej Besta, Simon Weber, Lukas Gianinazzi, Robert Gerstenberger, Andrey
  Ivanov, Yishai Oltchik, and Torsten Hoefler.
\newblock Slim graph: Practical lossy graph compression for approximate graph
  processing, storage, and analytics.
\newblock In {\em Proceedings of the International Conference for High
  Performance Computing, Networking, Storage and Analysis}, pages 1--25, 2019.

\bibitem{Blelloch:1996:PPA:227234.227246}
Guy~E. Blelloch.
\newblock Programming parallel algorithms.
\newblock {\em Commun. ACM}, 39(3):85--97, March 1996.

\bibitem{bouchitte2002listing}
Vincent Bouchitt{\'e} and Ioan Todinca.
\newblock Listing all potential maximal cliques of a graph.
\newblock {\em Theoretical Computer Science}, 276(1-2):17--32, 2002.

\bibitem{DBLP:journals/cacm/BronK73}
Coenraad Bron and Joep Kerbosch.
\newblock Finding all cliques of an undirected graph (algorithm 457).
\newblock {\em Commun. {ACM}}, 16(9):575--576, 1973.

\bibitem{buchanan2014solving}
Austin Buchanan, Jose~L. Walteros, Sergiy Butenko, and Panos~M. Pardalos.
\newblock Solving maximum clique in sparse graphs: an
  o(nm+n2\({}^{\mbox{d/4}}\) algorithm for d-degenerate graphs.
\newblock {\em Optim. Lett.}, 8(5):1611--1617, 2014.

\bibitem{cheng2011finding}
James Cheng, Yiping Ke, Ada Wai-Chee Fu, Jeffrey~Xu Yu, and Linhong Zhu.
\newblock Finding maximal cliques in massive networks.
\newblock {\em ACM Transactions on Database Systems (TODS)}, 36(4):1--34, 2011.

\bibitem{cheng2012fast}
James Cheng, Linhong Zhu, Yiping Ke, and Shumo Chu.
\newblock Fast algorithms for maximal clique enumeration with limited memory.
\newblock In {\em Proceedings of the 18th ACM SIGKDD international conference
  on Knowledge discovery and data mining}, pages 1240--1248, 2012.

\bibitem{Chiba1985}
Norishige Chiba and Takao Nishizeki.
\newblock Arboricity and subgraph listing algorithms.
\newblock {\em {SIAM} J. Comput.}, 14(1):210--223, 1985.

\bibitem{DBLP:journals/siamcomp/Cole88}
Richard Cole.
\newblock Parallel merge sort.
\newblock {\em {SIAM} J. Comput.}, 17(4):770--785, 1988.

\bibitem{cook2006mining}
Diane~J Cook and Lawrence~B Holder.
\newblock {\em Mining graph data}.
\newblock John Wiley \& Sons, 2006.

\bibitem{intelProducts}
Intel\textsuperscript{\textregistered} Corporation.
\newblock Intel\textregistered product specifications.
\newblock \url{https://ark.intel.com/content/www/us/en/ark.html#@Processors}.
\newblock Accessed: 2021-02-23.

\bibitem{Danisch2018}
Maximilien Danisch, Oana Balalau, and Mauro Sozio.
\newblock Listing k-cliques in sparse real-world graphs.
\newblock In Pierre{-}Antoine Champin, Fabien~L. Gandon, Mounia Lalmas, and
  Panagiotis~G. Ipeirotis, editors, {\em Proceedings of the 2018 World Wide Web
  Conference on World Wide Web, {WWW} 2018, Lyon, France, April 23-27, 2018},
  pages 589--598. {ACM}, 2018.

\bibitem{cscsWeb}
Centro~Svizzero di~Calcolo~Scientifico.
\newblock Cscs home page.
\newblock \url{https://www.cscs.ch}.
\newblock Acessed: 2021-02-23.

\bibitem{Downey1995}
Rodney~G. Downey and Michael~R. Fellows.
\newblock Fixed-parameter tractability and completeness {I:} basic results.
\newblock {\em {SIAM} J. Comput.}, 24(4):873--921, 1995.

\bibitem{Downey2013}
Rodney~G. Downey and Michael~R. Fellows.
\newblock {\em Fundamentals of Parameterized Complexity}.
\newblock Texts in Computer Science. Springer, 2013.

\bibitem{Eppstein2010}
David Eppstein, Maarten L{\"{o}}ffler, and Darren Strash.
\newblock Listing all maximal cliques in sparse graphs in near-optimal time.
\newblock In Otfried Cheong, Kyung{-}Yong Chwa, and Kunsoo Park, editors, {\em
  Algorithms and Computation - 21st International Symposium, {ISAAC} 2010, Jeju
  Island, Korea, December 15-17, 2010, Proceedings, Part {I}}, volume 6506 of
  {\em Lecture Notes in Computer Science}, pages 403--414. Springer, 2010.

\bibitem{galbrun2016top}
Esther Galbrun, Aristides Gionis, and Nikolaj Tatti.
\newblock Top-k overlapping densest subgraphs.
\newblock {\em Data Mining and Knowledge Discovery}, 30(5):1134--1165, 2016.

\bibitem{GareyNP-complete}
Michael~R. Garey and David~S. Johnson.
\newblock {\em Computers and Intractability; A Guide to the Theory of
  NP-Completeness}.
\newblock W. H. Freeman \& Co., USA, 1990.

\bibitem{gnuWeb}
Inc. GCC~Team, Free Software~Foundation.

\bibitem{jiang2013survey}
Chuntao Jiang, Frans Coenen, and Michele Zito.
\newblock A survey of frequent subgraph mining algorithms.
\newblock {\em The Knowledge Engineering Review}, 28(1):75--105, 2013.

\bibitem{lee2010survey}
Victor~E Lee, Ning Ruan, Ruoming Jin, and Charu Aggarwal.
\newblock A survey of algorithms for dense subgraph discovery.
\newblock In {\em Managing and Mining Graph Data}, pages 303--336. Springer,
  2010.

\bibitem{Snapnets}
Jure Leskovec and Andrej Krevl.
\newblock {SNAP Datasets}: {Stanford} large network dataset collection.
\newblock \url{http://snap.stanford.edu/data}, June 2014.

\bibitem{li2020ordering}
Rong-Hua Li, Sen Gao, Lu~Qin, Guoren Wang, Weihua Yang, and Jeffrey~Xu Yu.
\newblock Ordering heuristics for k-clique listing.
\newblock {\em Proceedings of the VLDB Endowment}, 13(12):2536--2548, 2020.

\bibitem{Lick1970}
Don~R. Lick and Arthur~T. White.
\newblock k-degenerate graphs.
\newblock {\em Canadian Journal of Mathematics}, 22(5):1082–1096, 1970.

\bibitem{Matula1983}
David~W. Matula and Leland~L. Beck.
\newblock Smallest-last ordering and clustering and graph coloring algorithms.
\newblock {\em J. {ACM}}, 30(3):417--427, 1983.

\bibitem{mitzenmacher2015scalable}
Michael Mitzenmacher, Jakub Pachocki, Richard Peng, Charalampos Tsourakakis,
  and Shen~Chen Xu.
\newblock Scalable large near-clique detection in large-scale networks via
  sampling.
\newblock In {\em Proceedings of the 21th ACM SIGKDD International Conference
  on Knowledge Discovery and Data Mining}, pages 815--824, 2015.

\bibitem{modani2008large}
Natwar Modani and Kuntal Dey.
\newblock Large maximal cliques enumeration in sparse graphs.
\newblock In {\em Proceedings of the 17th ACM conference on Information and
  knowledge management}, pages 1377--1378, 2008.

\bibitem{Nash-Williams1961}
C.~St.J.~A. Nash-Williams.
\newblock Edge-disjoint spanning trees of finite graphs.
\newblock {\em Journal of the London Mathematical Society}, s1-36(1):445--450,
  1961.

\bibitem{PAPADIMITRIOU1981131}
Christos~H. Papadimitriou and Mihalis Yannakakis.
\newblock The clique problem for planar graphs.
\newblock {\em Information Processing Letters}, 13(4):131--133, 1981.

\bibitem{rehman2012graph}
Saif~Ur Rehman, Asmat~Ullah Khan, and Simon Fong.
\newblock Graph mining: A survey of graph mining techniques.
\newblock In {\em Seventh International Conference on Digital Information
  Management (ICDIM 2012)}, pages 88--92. IEEE, 2012.

\bibitem{Reif:1993:SPA:562546}
John~H. Reif.
\newblock {\em Synthesis of Parallel Algorithms}.
\newblock Morgan Kaufmann Publishers Inc., San Francisco, CA, USA, 1st edition,
  1993.

\bibitem{DBLP:journals/jal/Robson86}
J.~M. Robson.
\newblock Algorithms for maximum independent sets.
\newblock {\em J. Algorithms}, 7(3):425--440, 1986.

\bibitem{NetworkRepository}
Ryan~A. Rossi and Nesreen~K. Ahmed.
\newblock The network data repository with interactive graph analytics and
  visualization.
\newblock In {\em AAAI}, 2015.

\bibitem{schmidt2009scalable}
Matthew~C Schmidt, Nagiza~F Samatova, Kevin Thomas, and Byung-Hoon Park.
\newblock A scalable, parallel algorithm for maximal clique enumeration.
\newblock {\em Journal of Parallel and Distributed Computing}, 69(4):417--428,
  2009.

\bibitem{shao2012managing}
Bin Shao, Haixun Wang, and Yanghua Xiao.
\newblock Managing and mining large graphs: systems and implementations.
\newblock In {\em Proceedings of the 2012 ACM SIGMOD International Conference
  on Management of Data}, pages 589--592, 2012.

\bibitem{Shi2020}
Jessica Shi, Laxman Dhulipala, and Julian Shun.
\newblock Parallel clique counting and peeling algorithms.
\newblock {\em CoRR}, abs/2002.10047, 2020.

\bibitem{DBLP:journals/kais/ShinEF18}
Kijung Shin, Tina Eliassi{-}Rad, and Christos Faloutsos.
\newblock Patterns and anomalies in k-cores of real-world graphs with
  applications.
\newblock {\em Knowl. Inf. Syst.}, 54(3):677--710, 2018.

\bibitem{tang2010graph}
Lei Tang and Huan Liu.
\newblock Graph mining applications to social network analysis.
\newblock In {\em Managing and Mining Graph Data}, pages 487--513. Springer,
  2010.

\bibitem{tomita2011efficient}
Etsuji Tomita, Tatsuya Akutsu, and Tsutomu Matsunaga.
\newblock Efficient algorithms for finding maximum and maximal cliques:
  Effective tools for bioinformatics.
\newblock In {\em Biomedical engineering, trends in electronics, communications
  and software}. IntechOpen, 2011.

\bibitem{DBLP:journals/tcs/TomitaTT06}
Etsuji Tomita, Akira Tanaka, and Haruhisa Takahashi.
\newblock The worst-case time complexity for generating all maximal cliques and
  computational experiments.
\newblock {\em Theor. Comput. Sci.}, 363(1):28--42, 2006.

\bibitem{DBLP:conf/www/Tsourakakis15a}
Charalampos~E. Tsourakakis.
\newblock The k-clique densest subgraph problem.
\newblock In {\em Proceedings of the 24th International Conference on World
  Wide Web, {WWW} 2015, Florence, Italy, May 18-22, 2015}, pages 1122--1132,
  2015.

\bibitem{DBLP:journals/gc/Wood07}
David~R. Wood.
\newblock On the maximum number of cliques in a graph.
\newblock {\em Graphs Comb.}, 23(3):337--352, 2007.

\bibitem{yuan2016diversified}
Long Yuan, Lu~Qin, Xuemin Lin, Lijun Chang, and Wenjie Zhang.
\newblock Diversified top-k clique search.
\newblock {\em The VLDB Journal}, 25(2):171--196, 2016.

\end{thebibliography}
